\documentclass[oneside,english]{amsart}
\usepackage[LGR,T1]{fontenc}
\usepackage[latin9]{inputenc}
\usepackage{amsthm}
\usepackage{amssymb}
\usepackage{esint}

\makeatletter

\DeclareRobustCommand{\greektext}{%
  \fontencoding{LGR}\selectfont\def\encodingdefault{LGR}}
\DeclareRobustCommand{\textgreek}[1]{\leavevmode{\greektext #1}}
\DeclareFontEncoding{LGR}{}{}
\DeclareTextSymbol{\~}{LGR}{126}

\numberwithin{equation}{section}
\numberwithin{figure}{section}
\theoremstyle{plain}
\newtheorem{thm}{\protect\theoremname}[section]
  \theoremstyle{plain}
  \newtheorem{cor}[thm]{\protect\corollaryname}
  \theoremstyle{plain}
  \newtheorem{prop}[thm]{\protect\propositionname}
  \theoremstyle{plain}
  \newtheorem*{prop*}{\protect\propositionname}
  \theoremstyle{remark}
  \newtheorem{rem}[thm]{\protect\remarkname}
  \theoremstyle{definition}
  \newtheorem{defn}[thm]{\protect\definitionname}
  \theoremstyle{plain}
  \newtheorem{lem}[thm]{\protect\lemmaname}
  \theoremstyle{definition}
  \newtheorem{example}[thm]{\protect\examplename}
  \theoremstyle{definition}
  \newtheorem{problem}[thm]{\protect\problemname}
  \theoremstyle{plain}
  \newtheorem{conjecture}[thm]{\protect\conjecturename}

\def\makebbb#1{
    \expandafter\gdef\csname#1\endcsname{
        \ensuremath{\Bbb{#1}}}
}\makebbb{R}\makebbb{N}\makebbb{Z}\makebbb{C}\makebbb{H}\makebbb{E}\makebbb{H}\makebbb{P}\makebbb{B}\makebbb{Q}\makebbb{E}

\usepackage{babel}

\usepackage{babel}

\makeatother

\usepackage{babel}

\makeatother

\usepackage{babel}
  \providecommand{\conjecturename}{Conjecture}
  \providecommand{\corollaryname}{Corollary}
  \providecommand{\definitionname}{Definition}
  \providecommand{\examplename}{Example}
  \providecommand{\lemmaname}{Lemma}
  \providecommand{\problemname}{Problem}
  \providecommand{\propositionname}{Proposition}
  \providecommand{\remarkname}{Remark}
\providecommand{\theoremname}{Theorem}

\begin{document}

\title{Large deviations for Gibbs measures with singular Hamiltonians and
emergence of Kähler-Einstein metrics}
\begin{abstract}
In the present paper and the companion paper \cite{berm8} a probabilistic
(statistical-mechanical) approach to the construction of canonical
metrics on a complex algebraic varieties $X$ is introduced, by sampling
''temperature deformed'' determinantal point processes. The main
new ingredient is a large deviation principle for Gibbs measures with
singular Hamiltonians, which is proved in the present paper. As an
application we show that the unique Kähler-Einstein metric with negative
Ricci curvature on a canonically polarized algebraic manifold $X$
emerges in the many particle limit of the canonical point processes
on $X.$ In the companion paper \cite{berm8} the extension to algebraic
varieties $X$ with positive Kodaira dimension is given and a conjectural
picture relating negative temperature states to the existence problem
for Kähler-Einstein metrics with positive Ricci curvature is developed. 
\end{abstract}

\author{Robert J. Berman}

\maketitle
\tableofcontents{}

\section{Introduction}

In the present paper and the companion paper \cite{berm8} a probabilistic
approach to the construction of canonical metrics on a complex algebraic
varieties $X$ is introduced, by sampling random point processes defined
in terms of algebro-geometric data, canonically attached to $X.$
The processes are ``positive temperature deformations'' of determinantal
(fermionic) point processes and the main new ingredient is a large
deviation principle for Gibbs measures with singular Hamiltonians
which is proved in the present paper. As an application we show that
the unique Kähler-Einstein metric with negative Ricci curvature on
a canonically polarized algebraic manifold $X$ emerges in the many
particle limit of the canonical point processes on $X.$ More generally,
in the presence of a stress-energy tensor on $X$ it is shown that
the unique Kähler metric solving Einstein's equation on $X$ with
negative cosmological constant (in Euclidean signature) emerges in
the many particle limit. 

The generalization to the construction of canonical metrics and measures
on a general algebraic variety $X$ of positive Kodaira dimension
are given in the companion paper \cite{berm8}, by exploiting the
global pluripotential theory and variational calculus in \cite{b-b,begz,bbgz,berm6}.
This leads to a new probabilistic link between algebraic geometry\emph{
}on one hand (in particular the Minimal Model Program) and Kähler-Einstein
geometry on the other. A conjectural picture is also developed describing
the relation between the existence of negative temperature states
and the existence problem for Kähler-Einstein metrics with positive
Ricci curvature. In particular, relations to algebro-geometric stability
properties, as in the Yau-Tian-Donaldson conjecture are described
in \cite{berm8}. See also \cite{berm7,hu} for connections to optimal
transport in the real setting (corresponding to the case when $X$
is toric and abelian variety, respectively) and \cite{berm5} for
connections to physics.

\subsection{A large deviation principle for Gibbs measures}

Let $X$ be a compact Riemannian manifold and denote by $dV$ the
corresponding volume form. Given a sequence of symmetric lower semi-continuous
functions $H^{(N)}$ on the $N-$fold products $X^{N}$ the corresponding
Gibbs measures at inverse temperature $\beta\in]0,\infty[$ is defined
as the following sequence of symmetric probability measures on $X^{N}:$
\[
\mu_{\beta}^{(N)}:=e^{-\beta H^{(N)}}dV^{\otimes N}/Z_{N,\beta},
\]
 where the normalizing constant 
\[
Z_{N,\beta}:=\int_{X^{N}}e^{-\beta H^{(N)}}dV^{\otimes N}
\]
is called the\emph{ ($N-$particle) partition function. }The ensemble
$(X^{N},\mu_{\beta}^{(N)})$ defines a random point process with $N$
particles on $X$ which, from the point of view of statistical mechanics,
models $N$ identical particles on $X$ interacting by the\emph{ Hamiltonian
}(interaction energy) $H^{(N)}$ in thermal equilibrium at \emph{inverse
temperature} $\beta.$ The corresponding \emph{empirical measure}
is the random measure 
\begin{equation}
\delta_{N}:\,\,X^{N}\rightarrow\mathcal{M}_{1}(X),\,\,\,(x_{1},\ldots,x_{N})\mapsto\delta_{N}(x_{1},\ldots,x_{N}):=\frac{1}{N}\sum_{i=1}^{N}\delta_{x_{i}}\label{eq:emp measure intro}
\end{equation}
 taking values in the space $\mathcal{M}_{1}(X)$ of all normalized
positive measures on $X,$ i.e. the space of all probability measures
on $X$. 

A classical problem is to establish conditions for the existence of
a macroscopic limit of the empirical measures $\delta_{N}$ in the
many particle limit $N\rightarrow\infty.$ More precisely, the problem
is to show that the random measures $\delta_{N}$ admit a deterministic
limit $\mu_{\beta}\in\mathcal{M}_{1}(X)$ in the sense that the law
\begin{equation}
\Gamma_{N}:=(\delta_{N})_{*}\mu_{\beta}^{(N)}\label{eq:law intro}
\end{equation}
of $\delta_{N},$ defining a probability measure on $\mathcal{M}_{1}(X),$
converges, as $N\rightarrow\infty,$ weakly to a Dirac mass concentrated
on some $\mu_{\beta}$ in $\mathcal{M}_{1}(X).$ Equivalently, the
marginals $(\mu_{\beta}^{(N)})_{j}$ of $\mu_{\beta}^{(N)}$ on $X^{j}$
satisfy 
\[
(\mu_{\beta}^{(N)})_{j}:=\int_{X^{N-j}}\mu_{\beta}^{(N)}\rightarrow\mu_{\beta}^{\otimes j},
\]
 weakly as probability measures on $X^{j}$ as $N\rightarrow\infty,$
which in the terminology of Kac and Snitzmann \cite{sn} means that
the sequence $\mu_{\beta}^{(N)}$ is\emph{ chaotic}. A stronger exponential
notion of convergence of $\delta_{N},$ with an explicit speed and
rate functional, is offered by the theory of large deviations, by
demanding that the laws $\Gamma_{N}$ satisfy a\emph{ Large Deviation
Principle (LDP)} with \emph{speed} $r_{N}$ and a \emph{rate functional}
$F,$ symbolically expressed as 
\[
\Gamma_{N}(\mu)\sim e^{-r_{N}F(\mu)},\,\,N\rightarrow\infty
\]
 and assuming that $F$ admits a unique minimizer $\mu_{\beta}$ in
$\mathcal{M}_{1}(X).$ Loosely speaking this means that the probability
of finding a cloud of $N$ points $x_{1},...,x_{N}$ on $X$ such
that the corresponding measure $\frac{1}{N}\sum_{i}\delta_{x_{i}}$
approximates a volume form $\mu$ is exponentially small unless $\mu$
is the minimizer $\mu_{\beta}$ of $F_{\beta}.$ 

Our main general result establish such a LDP for a class of singular
Hamiltonians:
\begin{thm}
\label{thm:gibbs intro}Let $H^{(N)}$ be a sequence of functions
(Hamiltonians) on $X^{N}$ as above\textup{. Assume that }
\begin{itemize}
\item \textup{there exists a sequence $\beta_{N}\rightarrow\infty$ of positive
numbers $\beta_{N}$ such that for any continuous function $u$ on
$X$ 
\[
\mathcal{F}_{\beta_{N}}(u):=-\frac{1}{N\beta_{N}}\log\int_{X^{N}}e^{-\beta_{N}\left(H^{(N)}(x_{1},...,x_{N})+u(x_{1})+...+u(x_{N})\right)}dV^{\otimes N}
\]
converges, as $N\rightarrow\infty,$ to a }Gateaux differentiable
functional $\mathcal{F}(u)$ on $C^{0}(X)$ 
\item $H^{(N)}$ is uniformly quasi-superharmonic, i.e. \emph{$\Delta_{x_{1}}H^{(N)}(x_{1},x_{2},...x_{N})\leq C$
on $X^{N}$}
\end{itemize}
Then, for any fixed $\beta>0,$ the measures $(\delta_{N})_{*}(e^{-\beta H^{(N)}}dV^{\otimes N})$
on $\mathcal{M}_{1}(X)$ satisfy, as $N\rightarrow\infty,$ a large
deviation principle (LDP) with \emph{speed} $\beta N$ and good \emph{rate
functional} 
\begin{equation}
F_{\beta}(\mu)=E(\mu)+\frac{1}{\beta}D_{dV}(\mu)\label{eq:free energy func theorem gibbs intro}
\end{equation}
 where the functional $E(\mu)$ is the Legendre-Fenchel transform
of $-\mathcal{F}(-\cdot)$ and $D_{dV}(\mu)$ is the entropy of $\mu$
relative to $dV.$ In particular, the empirical measures $\delta_{N}$
of the corresponding random point processes on $X$ converge in law
to the deterministic measure given by the unique minimizer $\mu_{\beta}$
of $F_{\beta}.$ Moreover, if the equation 
\begin{equation}
d\mathcal{F}_{|u}=\frac{e^{\beta u}dV}{\int_{X}e^{\beta u}dV}\label{eq:abstract mean field eq intro}
\end{equation}
on $C^{0}(X)$ admits a solution $u_{\beta},$ then the corresponding
differential $\mu_{\beta}:=d\mathcal{F}_{|u_{\beta}}$ is the minimizer
of $F_{\beta}.$ 
\end{thm}
It follows from the previous theorem that the LDP indeed also holds
for the corresponding Gibbs measures with the rate functional $F_{\beta}+C_{\beta},$
where $C_{\beta}$ is the following constant: 
\begin{equation}
C_{\beta}:=\inf_{\mathcal{M}_{1}(X)}F_{\beta}=-\lim_{N\rightarrow\infty}\frac{1}{N\beta_{N}}\log Z_{N.\beta_{N}},\label{eq:constant C and asympt of part intro}
\end{equation}
It should be stressed that even the convergence of the first marginals
of $\mu_{\beta}^{(N)},$ implied by the previous theorem, appears
to be a new result. 

As explained in Section \ref{sub:Relations-to-Gamma} the asymptotics
in the first assumption of the theorem may be replaced by the weaker
assumption that there exists a functional $E(\mu)$ on $\mathcal{M}_{1}(X)$
such that 
\[
H^{(N)}(x_{1},...,x_{N})/N\rightarrow E(\mu)
\]
in the sense of Gamma convergence. Moreover, Theorem \ref{thm:gibbs intro}
can be viewed as a generalization of the Gärtner-Ellis theorem in
the setting of Gibbs measures (see Section\ref{sub:Relations-to-the g-e}).
Let us also point out that that the restriction that $X$ be compact
can be removed if suitable growth-assumptions of $H^{(N)}$ ``at
infinity'' are made. But since our main application concerns the
case of compact complex manifolds, we have, for simplicity taken $X$
to be compact. 

It may be illuminating to point out that in thermodynamical terms
the content of Theorem \ref{thm:gibbs intro} can be heuristically
expressed as follows. Imagine that we know the macroscopic ground
state (i.e. the state of zero energy $E$) of a system of a large
number $N$ of particles in thermal equilibrium at zero temperature
(i.e. at $\beta=\infty)$. If we can rule out any \emph{first order
phase transitions} at zero-temperature (which essentially means that
the macroscopic equilibrium states is unique), then increasing the
temperature (i.e decreasing $\beta)$ leads to a new macroscopic equilibrium
state, minimizing the corresponding\emph{ free energy} functional
$E-S/\beta,$ where $S$ is the physical entropy (i.e. $S=-D$ with
our sign conventions). In fact, in the complex geometric setting to
which we next turn. the zero-temperature limit $\beta\rightarrow\infty$
is reminiscent of a (second order) gas-liquid phase transition \cite{berm9}.

\subsection{Application to Kähler-Einstein geometry}

Let now$X$ be an $n-$dimensional complex algebraic projective variety
of positive Kodaira dimension. This means that the plurigenera $N_{k}$
of $X$ are increasing: 
\[
N_{k}:=\dim_{\C}H^{0}(X,kK_{X})\rightarrow\infty,
\]
 where $H^{0}(X,kK_{X})$ denotes, as usual, the complex vector space
of all pluricanonical (holomorphic) $n-$forms of $X$ at level $k,$
i.e. $H^{0}(X,kK_{X})$ is the space of all global holomorphic sections
of the $k$ tensor power of the canonical line bundle 
\[
K_{X}:=\Lambda^{n}(T^{*}X)
\]
 of $X$ (using additive notation of tensor powers). In terms of local
holomorphic coordinates $z_{1},...,z_{n}$ on $X$ this simply means
that the elements $s^{(k)}$ of $H^{0}(X,kK_{X})$ may be represented
by local holomorphic functions $s^{(k)}$ on $X,$ such that $|s^{(k)}|^{2/k}$
transforms as a density on $X$ and thus defines a measure on $X.$
To any such algebraic variety $X$ we can associate the following
canonical sequence of probability measures $\mu^{(N_{k})}$ on $X^{N_{k}}:$
\begin{equation}
\mu^{(N_{k})}=:=\frac{1}{Z_{N_{k}}}\left|(\det S^{(k)})(z_{1},...,z_{N_{k}})\right|^{2/k},\label{eq:canon prob measure intro}
\end{equation}
 where $\det S^{(k)}$ is a generator of the top exterior power $\Lambda^{N_{k}}(H^{0}(X^{N_{k}},kK_{X^{N_{k}}}),$
i.e. totally antisymmetric (and thus defined up to a multiplicative
complex number) and $Z_{N_{k}}$is the normalizing constant. The probability
measure $\mu^{(N_{k})}$ thus defined is symmetric, i.e. invariant
under the natural action of the permutation group $S_{N_{k}},$ independent
of the choice of generator $\det S^{(k)}$ and hence defines a canonical
random point process on $X$ with $N_{k}$ points. 

As shown in the companion paper\cite{berm8}, it follows from Theorem
\ref{thm:gibbs intro}, combined with the asymptotics in \cite{b-b}
that the corresponding empirical measures $\delta_{N_{k}}$ converge
in law, as $k\rightarrow\infty,$ towards a deterministic measure
$\mu_{can}$ on $X,$ which is thus canonically attached to $X.$
In fact, using the pluripotential theory and variational calculus
in \cite{bbgz,berm6} the limiting measure $\mu_{can}$ is shown to
coincide with the canonical measure of Song-Tian \cite{s-t} and Tsuji
\cite{ts} previously defined in terms of Kähler-Einstein geometry
or equivalently as solutions to certain complex Monge-Ampère equations.
In the present paper we will show how to apply Theorem \ref{thm:gibbs intro}
in the special case when $K_{X}$ is positive (i.e. ample) to deduce
the following
\begin{thm}
\label{thm:ke intro}Let $X$ be a compact complex manifold such with
positive canonical line bundle $K_{X}.$ Then the empirical measures
$\delta_{N_{k}}$ of the corresponding canonical random point processes
on $X$ converge in law, as $N_{k}\rightarrow\infty,$ towards the
normalized volume form $dV_{KE}$ of the unique Kähler-Einstein metric
$\omega_{KE}$ on $X.$ More precisely, the law of $\delta_{N_{k}}$
satisfies a large deviation principle with speed $N_{k}$ whose rate
functional may be identified with Mabuchi's K-energy functional on
the space of Kähler metrics in $c_{1}(K_{X}).$ 
\end{thm}
By the celebrated Aubin-Yau theorem \cite{au,y} the canonical line
bundle $K_{X}$ of a compact complex manifold $X$ is positive precisely
when $X$ admits a Kähler-Einstein metric $\omega_{KE}$ with negative
Ricci curvature, i.e. a Kähler metric with constant negative Ricci
curvature: 
\begin{equation}
\mbox{Ric\,}\ensuremath{\omega_{KE}=-\omega_{KE}}\label{eq:k-e eq intro}
\end{equation}
However, there are very few examples where the Kähler-Einstein metric
can be obtained explicitly. The previous theorem provides a canonical
sequence of quasi-explicit Kähler forms $\omega_{k}$ approximating
$\omega_{KE}:$
\begin{cor}
\label{cor:ke intro}Let $X$ be a complex compact manifold such that
$K_{X}$ is positive. Then the sequence
\begin{equation}
\omega_{k}:=dd^{c}\log\int_{X^{N_{k}-1}}\left|(\det S^{(k)})(\cdot,x_{1},...,x_{N_{k}-1})\right|^{2/k}\label{eq:def of canonical sequ of current intro}
\end{equation}
(consisting of Kähler forms, for $k$ sufficiently large) converges,
as $k\rightarrow\infty,$ to the Kähler-Einstein metric $\omega_{KE}$
in the weak topology of currents on $X.$ 
\end{cor}
Theorem \ref{thm:ke intro} fits into a more general setting of ``temperature
deformed'' determinantal point processes attached to a polarized
manifold $(X,L),$ i.e. a compact complex manifolds $X$ endowed with
a positive line bundle $L$ (Theorem \ref{thm:def determ}). More
precisely, in the general setting the point processes are attached
to the data $(\left\Vert \cdot\right\Vert ,dV,\beta_{k})$ consisting
of a Hermitian metric $\left\Vert \cdot\right\Vert $ on a $L,$ a
volume form $dV$ on $X$ and a sequence of positive numbers $\beta_{k}\rightarrow\beta\in]0,\infty].$
Then the corresponding probability measures on $X^{N_{k}}$ are defined
by 
\begin{equation}
\mu^{(N_{k},\beta)}:=\frac{\left\Vert (\det S^{(k)})(x_{1},x_{2},...x_{N_{k}})\right\Vert ^{2\beta_{k}/k}dV^{\otimes N_{k}}}{Z_{N_{k},\beta}}\label{eq:prob measure for polarized intro}
\end{equation}
where $\det S^{(k)}$ is a generator of the top exterior power $\Lambda^{N_{k}}H^{0}(X,kL).$
Concretely, the corresponding LDP is equivalent to the following asymptotics
for the $L^{2\beta_{k}/k}-$norm of the generator $\det S^{(k)}$
of the determinant line of $H^{0}(X,kL)$ which is orthonormal with
respect to the $L^{2}-$product determined by $(\left\Vert \cdot\right\Vert ,dV):$
\[
\frac{1}{N_{k}}\log\left\Vert \det S^{(k)}\right\Vert _{L^{2\beta_{k}/k}(X^{N_{k}},\mu_{0}^{\otimes N_{k}})}\rightarrow-\inf_{\mathcal{\mu\in M}_{1}(X)}F_{\beta}(\mu)
\]
(by Lemma \ref{lem:appl of ge}). In this general setting the limiting
deterministic measure $\mu_{\beta}$ minimizing $F_{\beta}$ is the
volume form of the unique Kähler metric $\omega_{\beta}$ in the first
Chern class of $L$ solving the twisted Kähler-Einstein equation 
\begin{equation}
\mbox{Ric\,}\ensuremath{\omega=-\beta\omega+\eta},\label{eq:tw ke eq intro}
\end{equation}
where the twisting form $\eta$ is explicitly determined by $(\left\Vert \cdot\right\Vert ,dV,\beta).$
The point is that when $L=K_{X}$ any given volume form $dV$ naturally
defines a metric $\left\Vert \cdot\right\Vert _{dV}$ on $L$ and
the probability measures on $X^{N_{k}}$ attached to $(\left\Vert \cdot\right\Vert _{dV},dV,1)$
are precisely the canonical ones defined by formula \ref{eq:canon prob measure intro}.
Moreover, in this special case $\eta$ vanishes and the equation \ref{eq:tw ke eq intro}
thus reduces to the the usual Kähler-Einstein equation \ref{eq:k-e eq intro}.
The more general twisted version of the equation has previously appeared
in various situations in Kähler geometry \cite{f,s-t,ts}. From the
physics point of view the twisting form $\eta$ corresponds to the
(trace-reversed) stress-energy tensor in Einstein's equations on $X$
(with Euclidean signature).

The Hamiltonians 
\begin{equation}
H^{(N_{k})}(x_{1},...,x_{N_{k}}):=-k^{-1}\log\left\Vert (\det S^{(k)})(x_{1},x_{2},...x_{N_{k}})\right\Vert ^{2}\label{eq:ham cplx intro}
\end{equation}
corresponding to the probability measures \ref{eq:prob measure for polarized intro}
are strongly non-linear unless $X$ is a Riemann surface, i.e. unless
$n=1.$ In fact, in the simplest latter case, i.e. when $X$ is the
Riemann sphere, $H^{(N_{k})}(x_{1},...,x_{N_{k}})$ is a sum of identical
pair interactions $W(x_{i},x_{j})$, where $W$ is the Green function
of the corresponding Laplace operator and then the corresponding functional
$E(\mu)$ is the Dirichlet energy (Remark \ref{rem:energy}). In general,
the connection to the Kähler-Einstein geometry of $(X,L)$ will be
shown to arise from the fact that the equation \ref{eq:abstract mean field eq intro}
is intimately related to the complex Monge-Ampère equation 
\begin{equation}
(\omega_{0}+i\partial\bar{\partial}u)^{n}=e^{\beta u}dV\label{eq:a-y equation intro}
\end{equation}
 where $\omega_{0}$ is the normalized curvature two form of the given
metric $\left\Vert \cdot\right\Vert $ on $L.$ More precisely, the
two equations coincide for smooth functions $u$ such that $\omega_{0}+i\partial\bar{\partial}u$
is a Kähler form (i.e. smooth and positive). In this complex geometric
setting the strong non-linearity of the Hamiltonians $H^{(N)}$ when
$n\geq2$ is reflected in the non-linearity of the complex Monge-Ampère
operator appearing in the left hand side of equation \ref{eq:a-y equation intro}
(coinciding with the Laplacian when $n=1)$. Furthermore, the singularity
of $H^{(N)}$ (which is present for any dimension $n$) is a reflection
of the fact that solutions to the (generalized) Calabi-Yau equation
\begin{equation}
(\omega_{0}+i\partial\bar{\partial}u)^{n}=\mu\label{eq:c-y eq intro}
\end{equation}
are, in general, singular when $\mu$ is a probability measure on
$X$ (as is clear already for the Laplace equation appearing when
$n=1$). 

Finally, let us point out that the extension to general complex algebraic
manifolds $X$ with positive Kodaira dimension, established in the
companion paper \cite{berm8}, relies on an extension of Theorem \ref{thm:def determ}
to line bundles $L,$ which are big (but not necessarily positive);
see Section \ref{sub:The-generalization-to big}.

\subsection{Comparison with previous results}

First a comment on relations to the physics literature: in the case
$n=1$ (i.e. in two real dimensions) the quasi-linear Laplace type
equation \ref{eq:a-y equation intro} arises as the macroscopic equilibrium
equation in a range of statistical mechanical models of mean field
type: it is called the Joyce-Montgomery equation in Onsager's vortex
model for 2D turbulence, the Poisson-Boltzmann equation in the Debye-Hückel
theory of plasmas and electrolytes and the Lane-Emden equation in
stellar physics (see \cite{e-s}). But the Monge-Ampère equation ($n>1)$
does not seem to have a appeared in any statistical mechanical model
before. On the other hand, in the case when $\beta_{k}:=k$ the density
of the corresponding probability measure has a natural quantum mechanical
interpretation: it is the squared amplitude of the Slater determinant
representing a maximally filled many particle state of $N$ free fermions
on $X,$ subject to an exterior magnetic field (the corresponding
single particle wave functions are elements of $H^{0}(X,kL)$ and
represent the corresponding lowest Landau levels). The case when $\beta_{k}=\frac{1}{\nu}k,$
for a given positive integer $\nu,$ also appears in the fractional
Quantum Hall Effect, where the corresponding probability density is
the squared amplitude of the Laughlin state (see the review \cite{kl}
and references therein).

\subsubsection{Large deviations}

The LDP in Theorem \ref{thm:gibbs intro} in the case when $H^{(N)}$
is uniformly equicontinuous is essentially well-known in the setting
of mean field models \cite{e-h-t,berm7} (it then also applies to
the case of negative $\beta,$ by replacing $H^{(N)}$ with $-H^{(N)}).$
But the key feature of the previous theorem is that it applies to
a large class of \emph{singular} Hamiltonians and in particular $H^{(N)}$
is allowed to be strongly repulsive in the sense that it blows up,
as two points merge (and hence the Gibbs measure may be ill-defined
when $\beta$ is negative). It seems that the only previous class
where a convergence result as in Theorem \ref{thm:g-e} has been established
for singular Hamiltonians is in the ``linear'' case when $H^{(N)}$
is a sum of pair interactions with a mean field scaling: 
\begin{equation}
H^{(N)}(x_{1},....,x_{N})=\frac{1}{(N-1)}\sum_{1\leq i<j\leq N}W(x_{i},x_{j}),\label{eq:mean field pair hamilton}
\end{equation}
 where the pair interaction $W$ is allowed to be singular along the
diagonal, as long it is lower semi-continuous and in $L_{loc}^{1}$
(this is indeed a mean field interaction in the sense that each particle
$x_{i}$ is exposed to the average of the pair interactions $W(x_{i},x_{j})$
for the $N-1$ remaining particles). Then the asymptotics of the partitions
functions \ref{eq:constant C and asympt of part intro} can be obtained
using the method of Messer-Spohn\cite{m-s}, which is based on the
Gibbs variational principle and which crucially relies on the the
existence of the\emph{ mean energy} $\bar{E}(\mu)$ corresponding
to $H^{(N)}$ (see \cite{k,clmp} for the case of a logarithmic singularity
which is motivated by Onsager's vortex model for 2D turbulence \cite{o,e-s}).
A similar argument applies in the case of ``finite order'', i.e.
when $H^{(N)}$ is a sum of $j-$point interactions for a uniformly
bounded $j$ (then $E(\mu)$ depends polynomially on $\mu)$. However,
the main point of the previous theorem is to avoid the latter assumption
which is not satisfied in the application to Kähler-Einstein geometry
(apart from the classical lowest dimensional setting of Riemann surfaces).
In particular, the present proof bypasses the problem of the existence
of the limiting mean energies. Instead the main idea of the proof
is to exploit the Riemannian orbifold geometry of the space of configurations
of $N$ points on $X,$ viewed as the singular quotients $X^{N}/S_{N},$
where $S_{N}$ is the symmetric group acting on $X^{N}$ by permuting
the factors. The key result is a submean inequality for positive quasi-subharmonic
functions on $X^{N}/S_{N}$ with a distortion coefficient which is
sub-exponential in the dimension (Theorem \ref{thm:submean ineq text}),
which is closely related to an inequality of Li-Schoen \cite{li-sc}.

There is also another approach to large deviation principles for mean
field Hamiltonians of the form \ref{eq:mean field pair hamilton}
originating in the literature on random matrices and Coulomb gases
\cite{ben-g,be-z,c-g-z,se}, which as explained in \cite{se}, is
closely related to the notion of Gamma convergence (see also \cite{z-z,z2}
for applications to univariat random polynomials). This approach seems
to be limited to the case when $\beta_{N}\gg\log N$ and in particular
$\beta=\infty$ so that the entropy contributions can be neglected.\footnote{The Hamiltonians in the random matrix and Coulomb gas literature are
usually scaled in a different way so that our zero-temperature $(\beta=\infty)$
corresponds to a fixed inverse temperature. }. See also \cite{d-l-r} for a general LDP for Hamiltonians of the
form \ref{eq:mean field pair hamilton} using weak convergence methods. 

Let us also point out that the role of $(\det S^{(k)})(x_{1},x_{2},...x_{N_{k}})$
appearing in formula \ref{eq:prob measure for polarized intro} is
played by the classical Vandermonde determinant in the random matrix
literature (for Example \ref{exa:proj subm}. In fact, there is a
non-compact analogue of Theorem \ref{thm:def determ} in Euclidean
$\C^{n}$ which specializes to the setting of random matrix theory
and the 2D log gas when $n=1$ and $\beta=\infty$ and to the 2D vortex
model (for $n=1$ and $\beta<\infty)$ and which can be proved by
supplementing the proof of Theorem \ref{thm:def determ} with a tightness
estimate, as in the non-compact setting considered for $\beta=\infty$
in \cite{berm 1 komma 5} (see also \cite{b-l} for the case $\beta=\infty$).
Details will appear elsewhere.

\subsubsection{Kähler geometry}

A statistical mechanics approach has previously been applied to conformal
geometry \cite{k2}, as opposed to the present complex-geometric setting.
The role of the ``determinantal'' Hamiltonian \ref{eq:ham cplx intro}
is in the conformal setting played by a mean field Hamiltonian of
the form \ref{eq:mean field pair hamilton} with a logarithmic pair
interaction and the role of the fully non-linear complex Monge-Ampère
operator is played by a linear conformally invariant operator, which
is zero-order perturbation of a power of the Laplacian (the Paneitz
operator). Accordingly previous results in \cite{k,clmp} concerning
such Hamiltonians can be applied in the conformal setting (compare
the discussion above) in the conformal setting, while the present
setting seems to require new methods. 

The present probabilistic should be viewed in the light of the pervasive
philosophy in Kähler geometry, going back to Yau \cite{y2}, of approximating
metrics on a complex algebraic manifold with algebraically defined
Bergman metrics, which may me identified with elements of the symmetric
space $GL(N,\C)/U(N).$ For example, the quasi-explicit Kähler metrics
$\omega_{k}$ in formula \ref{eq:def of canonical sequ of current intro},
approximating the Kähler-Einstein metric $\omega_{KE}$ on a canonically
polarized manifold $X,$ are analogs of Donaldson's balanced metrics
in $GL(N,\C)/U(N)$ \cite{do3}. One advantage of the present approach
is that, as shown in the companion paper \cite{berm8}, the approximation
also applies when, for example, $X$ is of general type, where the
role of $\omega_{KE}$ is played by the the canonical Kähler-Einstein
current on $X$ (which is singular along a subvariety of $X)$ \cite{begz,bbgz}.
In another direction it would be interesting to see if the present
approach can be implemented to construct numerical simulations of
Kähler-Einstein metrics, using Monte Carlo type methods, complementing
the different numerical approaches in \cite{do3,d-h-h-k} (see \cite{b-h}
for relations between Monte Carlo simulations and similar polynomial
determinantal point processes). 

Even if the connection between canonical random point processes on
a complex algebraic manifold $X$ does not seem to have been studied
before, there are some connections to previous work on random polynomials/holomorphic
sections in a given back-ground geometry \cite{s-z}; in particular
in the one-dimensional setting where an LDP was obtained in \cite{z-z,z2}.
Another probabilistic approach to the space of Kähler metrics has
been introduced in a a series of papers by Ferrari, Klevtsov and Zelditch\cite{f-k-z},
motivated by Quantum Field Theory. The approach aims at approximating
random random Kähler metrics with random Bergman metrics. Accordingly,
the role of the $N-$particle space $X^{N}/S_{N}$ is in \cite{f-k-z}
played by the symmetric space $GL(N.\C)/U(N).$ In conclusion, it
would be very interesting to understand the precise connections between
\cite{f-k-z} and the present setting, as well as the connection to
Donaldson's balanced metrics \cite{do3}.

\subsubsection*{Acknowledgment}

It is a pleasure to thank Sebastien Boucksom, David Witt-Nyström,
Vincent Guedj and Ahmed Zeriahi for the stimulating collaborations
\cite{b-b,b-b-w,bbgz}, which paved the way for the present work.
I am also grateful to Bo Berndtsson for infinitely many fruitful discussions
on complex analysis and Kähler geometry over the years. Thanks, in
particular, to Sebastien Boucksom for illuminating discussions on
Lemma \ref{lem:leg-fench approx}. The present paper, together with
the companion paper \cite{berm8}, supersedes the first arXiv version
of the paper \cite{berm8}. This work was supported by grants from
the ERC and the KAW foundation.

\subsubsection*{Organization}

In section \ref{sec:Submean-inequalities-in} we prove the submean
inequality in large dimensions, which plays a key role in the subsequent
section \ref{sec:A-large-deviation} where the general LDP in Theorem
\ref{thm:gibbs intro} is proved. In Section \ref{sec:Relations-to-convergence,}
we make a digression on relations to previous methods and notions
used in the literature on large deviations. The applications to Kähler-Einstein
geometry are given and Section \ref{sec:Applications-to-K=0000E4hler-Einstein}.
For the convenience of readers lacking background in Kähler geometry
we start the section by giving a reasonably self-contained account
of the Kähler geometry setup (including some rudiments of pluripotential
theory). The article is concluded with an outlook in Section \ref{sec:Outlook}
on some open problems and an appendix where the dimension dependence
on the constant in the Cheng-Yau gradient estimate is obtained, by
tracing through the usual proof.

\section{\label{sec:Submean-inequalities-in}Submean inequalities in large
dimension }

\subsection{Setup}

Let $(X,g)$ be a $n-$dimensional Riemannian manifold and assume
that 
\[
\mbox{Ric \ensuremath{g\geq-\kappa^{2}(n-1)g}}
\]
 for some positive constant $\kappa$ (sometimes referred to as the
normalized lower bound on the Ricci curvature). Let $G$ a finite
group acting by isometries on $X$ and denote by $M:=X/G$ the corresponding
quotient equipped with the distance function induced by the metric
$g,$ i.e. 
\[
d_{M}(x,y):=\inf_{\gamma\in G}d_{X}(x,\gamma y),
\]
 where $d_{X}$ is the Riemannian distance function on $(X,g).$ Even
though the quotient $M$ is not a manifold in general (since $G$
will in general have fixed points) it still comes with a smooth structure
in the following sense. Denote by $p$ the natural projection map
from $X^{N}$ to $M.$ Using the projection $p$ we can identify a
function $f$ on $M$ with $G-$invariant function $p^{*}f$ on $X$
and accordingly we say that $f$ is smooth if $p^{*}f$ is. Similarly,
there is a natural notion of Laplacian $\Delta$on the quotient $M:$
the Laplacian $\Delta u$ of a locally integrable function $u$ on
$M$ is the signed Radon measure defined by 
\[
\int_{M}(\Delta u)f:=\frac{1}{|G|}\int_{X}p^{*}u\Delta(p^{*}f)
\]
 for any smooth function $f$ on $M.$ More generally, by localization,
this setup naturally extends to the setting of Riemannian orbifolds
(see \cite{B-}), but the present setting of global quotients will
be adequate for our purposes.

\subsection{Statement of the submean inequality}
\begin{thm}
\label{thm:submean ineq text}Let $(X,g)$ be a Riemannian manifold
of dimension $n$ such that $\mbox{Ric \ensuremath{g\geq-\kappa^{2}(n-1)g}}$
and $G$ a finite group acting by isometries on $X.$ Denote by $M:=X/G$
the corresponding quotient equipped with the distance function induced
by the metric $g$ and let $v$ be a non-negative function on $M$
such that $\Delta_{g}v\geq-\lambda^{2}v$ for some non-negative constant
$\lambda.$ Then, for any $\delta>0$ and $\epsilon\in]0,1]$ there
exist constants $A$ and $C$ such that 
\[
\sup_{B_{\epsilon\delta}(x_{0})}v^{2}\leq Ae^{2\lambda\delta}e^{Cn(\delta+\epsilon)}\frac{\int_{B_{\delta}(x_{0})}v^{2}dV}{\int_{B_{\epsilon\delta}(x_{0})}dV},
\]
 where $C$ only depends on an upper bound on $\kappa$ and $A$ only
depends on $\delta$ and $\epsilon$ (assuming that the balls above
are contained in a compact subset of $M).$ 
\end{thm}
Note that by the $G-$invariance we may as well replace the functional
$v$ and the balls on $M$ with their pull-back to $X.$

\subsection{Proof of the submean inequality in Theorem \ref{thm:submean ineq text}}

We will follow closely the elegant proof of Li-Schoen \cite{li-sc}
of a similar submean inequality. But there are two new features here
that we have to deal with: 
\begin{itemize}
\item We have to make explicit the dependence on the dimension $n$ of all
constants and make sure that the final contribution is sub-exponential
in $n$ 
\item We have to adapt the results to the singular setting of a Riemannian
quotient
\end{itemize}
Before turning to the proof we point out that it is well-known that
submean inequalities with a multiplicative constant $C(n)$ do hold
in the more general singular setting of Alexandrov spaces (with a
strict lower bound $-\kappa$ on the sectional curvature). But it
seems that the current proofs (see for example \cite{h-x}), which
combine local Poincaré and Sobolev inequalities with the Moser iteration
technique, do not give the subexponential dependence on $C(n)$ that
we need. 

We recall that the two main ingredients in the proof of the result
of Li-Schoen referred to above is the gradient estimate of Cheng-Yau
\cite{c-y} and a Poincaré-Dirichlet inequality on balls. Let us start
with the gradient estimate that we will need:
\begin{prop}
\label{prop:cheng-yau}Let $u$ be a harmonic function on $B_{a}(x_{0})$
in $M.$ Set $\rho_{x_{0}}(x):=d(x,x_{0})$ (the distance between
$x$ and $x_{0}).$ Then 
\[
\sup_{B_{a}(x_{0})}\left(\left|\nabla\log u\right|(a-\rho_{x_{0}})\right)\leq Cn(1+\kappa a)\,\,\,\,\,(C_{n}\leq Cn)
\]
 for some absolute constant $C$ (in particular, independent of $n,$
$\kappa$ and $a).$ \end{prop}
\begin{proof}
In the smooth case this is the celebrated Cheng-Yau gradient estimate
\cite{c-y}. The result is usually stated without an explicit estimate
of the multiplicative constant $C_{n}$ in terms of $n,$ but tracing
through the proof in \cite{c-y} gives $C_{n}\leq Cn$ (see the appendix
in the present paper and also \cite{a-d-b} for a probabilistic proof
providing an explicit constant). We claim that the same estimate holds
in the present setting using a lifting argument. To see this recall
that the usual proof of the gradient estimate proceeds as follows
(see the appendix). Set $\phi(x):=\left|\nabla\log u\right|(=\left|\nabla u\right|/u)$
and $F(x):=\phi(x)(\rho_{x_{0}}-a)^{2}.$ Then $F$ attains its maximum
in a point $x_{1}$ in the interior of $B_{a}(x_{0})$ (otherwise
$\left|\nabla u\right|$ vanishes identically and then we are trivially
done). Hence, $F(x)\leq F(x_{1})$ on some neighborhood $U$ of $x_{1}.$
Now, in case $F$ (or equivalently $\rho_{x_{0}})$ is smooth on $U$
we get $\Delta F\leq0$ and $\nabla F=0$ at $x_{1}.$ Calculating
$\Delta F$ and using Bochner formula and Laplacian comparison then
gives 
\begin{equation}
\phi(x_{1})(a-\rho_{x_{0}}(x_{1}))\leq Cn(1+\kappa a)\label{eq:proof of cheng-yau}
\end{equation}
 which is the desired estimate. In the case when $\rho_{x_{0}}$ is
not smooth on $U,$ i.e. $x_{1}$ is contained in the cut locus of
$x_{0}$ one first replaces $\rho_{x_{0}}$ with a smooth approximation
$\rho_{x_{0}}^{(\epsilon)}$ of $\rho_{x_{0}}$ (which is a local
barrier for $\rho_{x_{0}})$ and then lets $\epsilon\rightarrow0$
to get the same conclusion as before. In the singular case $M=X/G$
we proceed as follows. First we identify $F$ with a $G-$invariant
function on the inverse image of $B_{R}(x_{0})$ in $X$ (and $x_{0}$
and $x_{1}$ with a choice of lifts in the corresponding $G-$orbits)
and set $\tilde{F}:=(x)(a-\tilde{\rho}_{x_{0}})^{2},$ where $\tilde{\rho}_{x_{0}}(x):=d_{X}(x_{0},x).$
By definition $\tilde{\rho}_{x_{0}}\geq\rho_{x_{0}}$ on $X$ and,
after possibly changing the lift of the point $x_{1}$ we may assume
that $\tilde{\rho}_{x_{0}}=\rho_{x_{0}}$ at $x=x_{1}$ and hence
$\tilde{\rho}_{x_{0}}<a$ (after perhaps shrinking $U).$ In particular,
$\tilde{F}\leq F$ on $U$ and $\tilde{F}=F$ at $x_{1}$ and hence
$\tilde{F}$ also has a local maximum at $x_{1}.$ But then the previous
argument in the smooth case gives that \ref{eq:proof of cheng-yau}
holds with $\rho_{x_{0}}$ replaced by $\tilde{\rho}_{x_{0}}.$ But
since the two functions agree at $x_{1}$ this concludes the proof
in the general case.\end{proof}
\begin{cor}
\label{cor:harnack type}Let $h$ be a positive harmonic function
on $B_{\delta}(x_{0}).$ Then there exists a constant $C$ only depending
on an upper bound on $\kappa$ such that 
\[
\sup_{B_{\epsilon\delta}(x_{0})}h^{2}\leq e^{C\epsilon n}\frac{\int_{B_{\epsilon\delta}(x_{0})}h^{2}dV}{\int{}_{B_{\epsilon\delta}(x_{0})}dV}
\]
for $0<\epsilon<1.$\end{cor}
\begin{proof}
Set $v:=\log h$ and fix $x\in B_{\epsilon\delta}(x_{0}).$ Integrating
along a minimizing geodesic connecting $x_{0}$ and $x$ and using
the gradient estimate in the previous proposition gives 
\[
\left|v(x)-v(x_{0})\right|\leq Cn\int_{0}^{\epsilon\delta}dt\frac{1}{\delta-t}dt=Cn\left(\log(\delta-0)-\log(\delta-\epsilon\delta)\right)=-Cn\log(1-\epsilon).
\]
 In particular, for any two points $x,y\in B_{\epsilon\delta}(x_{0})$
we get $\left|v(x)-v(y)\right|\leq\left|v(x)-v(x_{0})\right|+\left|v(y)-v(x_{0})\right|\leq-2Cn\log(1-\epsilon),$
i.e. $h(x)\leq(1-\epsilon)^{-2Cn}h(y).$ In particular, $\sup_{B_{\epsilon\delta}(x_{0})}h^{2}\leq(1-\epsilon)^{-4Cn}\inf_{B_{\epsilon\delta}(x_{0})}h^{2},$
which implies the proposition after renaming the constant $C.$
\end{proof}
The second key ingredient in the proof of Theorem \ref{thm:submean ineq text}
is the following Poincaré-Dirichlet inequality:
\begin{prop*}
\label{prop: poincare}Let $f$ be a smooth function on $B_{\delta}(x_{0})$
vanishing on the boundary. Then

\[
\int_{B_{\delta(x_{0})}}|f|^{2}dV_{g}\leq4e^{Cn\delta}\int_{B_{\delta(x_{0})}}|\nabla f|^{2}dV_{g}
\]
 where the constant $C$ only depends on an upper bound on $\kappa.$\end{prop*}
\begin{proof}
We follow the proof in \cite{li-sc} with one crucial modification
(compare the remark below). To fix ideas we first consider the case
of a Riemannian manifold. Fix a point $p$ in the boundary of the
ball $B_{1}(x_{0})$ and denote by $r_{1}(x)$ the distance between
$x\in M$ and $p.$ From the standard comparison estimate for the
Laplacian we get 
\begin{equation}
\Delta r_{1}\leq(n-1)(\frac{1}{r_{1}}+\kappa)\label{eq:laplacian comparis}
\end{equation}
(in the weak sense and point-wise away from the cut locus of $p$).
In particular, for any positive number $a$ we deduce the following
inequality on $B_{\delta}(x_{0})$ (using that $g(\nabla r_{1},\nabla r_{1})=1)$
a.e.) 
\[
\Delta_{g}(e^{-ar_{1}})=ae^{-ar_{1}}(a-\Delta r_{1})\geq ae^{-a(1+\delta)}\left(a-(n-1)(\frac{1}{(1-\delta)}+\kappa)\right)
\]
Hence, setting $a:=n(\frac{1}{(1-\delta)}+\kappa)$ gives 
\[
\Delta_{g}(e^{-ar_{1}})\geq ae^{-a(1+\delta)}(\frac{1}{(1-\delta)}+\kappa)>0
\]
Multiplying by $\left|f\right|$ and integrating once by parts (and
using that $\left\Vert \nabla r_{1}\right\Vert \leq1$ ) we deduce
that 
\[
a\int_{B_{\delta(x_{0})}}|\nabla f|e^{-ar_{1}}dV\geq a(\frac{1}{(1-\delta)}+\kappa)\int_{B_{\delta(x_{0})}}|f|e^{-a(1+\delta)}dV)
\]
Estimating $e^{-ar_{1}}\leq e^{-a(1-\delta)}$ in the left hand side
above and rearranging gives 
\[
\int_{B_{\delta(x_{0})}}|\nabla f|dVe^{2a\delta}(\frac{1}{(1-\delta)}+\kappa)^{-1}\geq\int_{B_{\delta(x_{0})}}|f|dV,
\]
 (using that $g(\nabla r_{1},\nabla r_{1})\leq1$ in the sense of
upper gradients). This shows that the $L^{1}-$version of the Poincaré
inequality in question holds with the constant $(\frac{1}{(1-\delta)}+\kappa)^{-1}e^{\delta2n(\frac{1}{(1-\delta)}+\kappa)},$
which for $\delta$ sufficiently small is bounded from above by $e^{n(4+2\kappa)\delta}.$
The general Riemannian $L^{2}-$Poincare inequality now follows from
replacing $|f|$ with $|f|^{2}$ and using Hölder's inequality. Finally,
in the case of the a Riemannian quotient $M$ we can proceed exactly
as above using that the Laplacian comparison estimate in formula \ref{eq:laplacian comparis}
is still valid. Indeed, the pull-back $p^{*}r_{1}$ of $r_{1}$ to
$X$ is an infimum of functions for which the corresponding estimate
holds (by the usual Laplacian comparison estimate and the assumption
that $G$ acts by isometries). But then the estimate also holds for
the function $p^{*}r_{1},$ by basic properties of Laplacians. More
generally, the required Laplacian comparison estimate was shown in
\cite{B-} for general Riemannian orbifolds.\end{proof}
\begin{rem}
The only difference from the argument used in \cite{li-sc} is that
we have taken the point $p$ to be of distance $1$ from $x_{0}$
rather than distance $2\delta,$ as used in \cite{li-sc}. For $\delta$
small this change has the effect of improving the exponential factor
from $e^{n(1+\delta\kappa)}$ to $e^{n(\delta+\delta\kappa)},$ which
is crucial as we need a constant in the Poincare inequality which
has subexponential growth in $n$ as $\delta\rightarrow0.$ 
\end{rem}

\subsubsection{End of proof of Theorem \ref{thm:submean ineq text}}

Let us first consider the case when $\lambda=0.$ Denote by $h$ the
harmonic function on $B_{\delta}$ coinciding with $v$ on $\partial B_{\delta}.$
By Cor \ref{cor:harnack type} and the subharmonicity of $v$ 
\[
\sup_{B_{\epsilon\delta}(x_{0})}v^{2}\leq e^{Cn\epsilon}\frac{\int_{B_{\epsilon\delta}(x_{0})}|h|^{2}dV_{g}}{\int_{B_{\epsilon\delta}(x_{0})}dV_{g}}.
\]
 Next, by the triangle inequality
\[
\int_{B_{\epsilon\delta}(x_{0})}|h|^{2}dV_{g}/2\leq\int_{B_{\epsilon\delta}(x_{0})}|h-v|^{2}dV+\int_{B_{\epsilon\delta}(x_{0})}|v|^{2}dV
\]
Since $h-v$ vanishes on the boundary of $B_{\delta}(x_{0})$ applying
the Poincare inequality in Prop \ref{prop: poincare} then gives $\int_{B_{\epsilon\delta}(x_{0})}|h-v|^{2}dV\leq$
\[
\leq\int_{B_{\delta(x_{0})}}|h-v|^{2}dV\leq Ae^{Bn\delta}\int_{B_{\delta(x_{0})}}|\nabla h-\nabla v|^{2}dV\leq2Ae^{Bn\delta}\int_{B_{\delta(x_{0})}}|\nabla h|^{2}+|\nabla v|^{2}dV
\]
But $h$ is the solution to a Dirichlet problem and as such minimizes
the Dirichlet norm $\int_{B_{\delta(x_{0})}}|\nabla h|^{2}$ over
all subharmonic functions with the same boundary values as $h.$ Accordingly,
\[
\int_{B_{\epsilon\delta}(x_{0})}|h-v|^{2}dV\leq Ae^{Bn\delta}\int_{B_{\delta(x_{0})}}|\nabla v|^{2}dV
\]
 Finally, using that $v$ is subharmonic we get 
\[
\int_{B_{\delta(x_{0})}}|\nabla v|^{2}dV\leq C_{\delta}\int_{B_{2\delta(x_{0})}}|v|^{2}dV
\]
 (as is seen by multiplying with a suitable smooth function $\chi$
supported on $B_{2\delta}$ such that $\chi=1$ on $B_{\delta}).$
All in all this concludes the proof of Theorem \ref{thm:submean ineq text}
in the case $\lambda=0.$

Finally, to handle the general case (i.e. $\lambda\neq0)$ we set
$N:=M\times]-1,1[$ equipped with the standard product metric and
apply the previous case to the function $ve^{\lambda t}$ to get 
\[
\sup_{B_{\epsilon\delta}(x_{0},0)\subset N}v^{2}e^{2\lambda t}\leq A_{\delta}e^{Bn(\delta+\epsilon)}\frac{\int_{B_{2\delta}(x_{0},0)\subset N}v^{2}e^{2\lambda t}dV}{\int_{B_{\epsilon\delta}(x_{0},0)\subset N}dV},
\]
But restricting the sup in the left hand side to $B_{\epsilon\delta}(x_{0})\times\{0\}$
and using that $B_{\epsilon\delta/2}(x_{0},0)\times[-\epsilon\delta/2,\epsilon\delta/2]\subset B_{\epsilon\delta}(x_{0},0)$
and $B_{2\delta}(x_{0},0)\subset B_{2\delta}(x_{0},0)\times[2\delta,2\delta]$
gives 
\[
\sup_{B_{\epsilon\delta}(x_{0})\subset M}v^{2}\leq A_{\delta,\epsilon}e^{2\lambda\delta}e^{Bn(\delta+\epsilon)}\frac{\int_{B_{2\delta}(x_{0})\subset M}v^{2}dV_{g}}{\int_{B_{\epsilon\delta/2}(x_{0},0)\subset M}dV_{g}},
\]
which concludes the proof of the general case (after a suitable rescaling).

\section{\label{sec:A-large-deviation}Proof of the large deviation principle
for Gibbs measures }

Given a compact topological space $X$ we will denote by $C^{0}(X)$
the space of all continuous functions $u$ on $X,$ equipped with
the sup-norm and by $\mathcal{M}(X)$ the space of all signed (Borel)
measures on $X.$ The subset of $\mathcal{M}(X)$ consisting of all
probability measures will be denoted by $\mathcal{M}_{1}(X).$ We
endow $\mathcal{M}(X)$ with the weak topology, i.e. $\mu_{j}$ is
said to converge to $\mu$ weakly in $\mathcal{M}(X)$ if 
\[
\left\langle \mu_{j},u_{j}\right\rangle \rightarrow\left\langle \mu,u\right\rangle 
\]
 for any continuous function $u$ on $X,$ i.e. for any $u\in C^{0}(X),$
where $\left\langle u,\mu\right\rangle $ denotes the standard integration
pairing between $C^{0}(X)$ and $\mathcal{M}(X)$ (equivalently, the
weak topology is precisely the weak{*}-topology when $\mathcal{M}(X)$
is identified with the topological dual of $C^{0}(X)).$ A functional
$\mathcal{F}$ on $C^{0}(X)$ will be said to be \emph{Gateaux differentiable
}if it is differentiable along affine lines and for any $u$ in $C^{0}(X)$
there exists an element $d\mathcal{F}_{|u}$ in $\mathcal{M}(X),$
called the differential of $\mathcal{F}$ at $u,$ such that for any
$v$ in $C^{0}(X)$ 
\[
\frac{d\mathcal{F}(u+tv)}{dt}_{|t=0}=\left\langle d\mathcal{F}_{|u},v\right\rangle 
\]

\subsection{\label{sub:Setup:-the-Boltzmann-Gibbs}Setup: the Gibbs measure $\mu_{\beta}^{(N)}$
associated to the Hamiltonian $H^{(N)}$}

A\emph{ random point process} with $N$ particles is by definition
a probability measure $\mu^{(N)}$ on the $N-$particle space $X^{N}$
which is symmetric, i.e. invariant under permutations of the factors
of $X^{N}.$ The\emph{ empirical measure} of a given random point
process is the following random measure 
\begin{equation}
\delta_{N}:\,\,X^{N}\rightarrow\mathcal{M}_{1}(X),\,\,\,\mapsto(x_{1},\ldots,x_{N})\mapsto\delta_{N}(x_{1},\ldots,x_{N}):=\frac{1}{N}\sum_{i=1}^{N}\delta_{x_{i}}\label{eq:empirical measure text}
\end{equation}
on the ensemble $(X^{N},\mu^{(N)}).$ By definition the\emph{ law}
of $\delta_{N}$ is the push-forward of $\mu^{(N)}$ to $\mathcal{M}_{1}(X)$
under the map $\delta_{N},$ which thus defines a probability measure
on $\mathcal{M}_{1}(X).$ 

Now fix a back-ground measure $\mu_{0}$ on $X$ and let $H^{(N)}$
be a given\emph{ $N-$particle Hamiltonian,} i.e. a symmetric function
on $X^{N},$ which we will assume is lower semi-continuous (and in
particular bounded from below, since $X$ is assumed compact). Also
fixing a positive number $\beta$ the corresponding\emph{ Gibbs measure}
(at inverse temperature $\beta)$ is the symmetric probability measure
on $X^{N}$ defined as 
\[
\mu_{\beta}^{N}:=e^{-\beta H^{(N)}}\mu_{0}^{\otimes N}/Z_{N},
\]
 where the normalizing constant 
\[
Z_{N,\beta}:=\int_{X^{N}}e^{-\beta H^{(N)}}\mu_{0}^{\otimes N}
\]
is called the\emph{ ($N-$particle) partition function. }In our setting
we will take $\mu_{0}$ to be the volume form $dV$ of a fixed Riemannian
metric. Given a continuous function $u$ on $X$ we will also write
\[
Z_{N,\beta}[u]:=\int_{X^{N}}e^{-\beta(H^{(N)}+u)}\mu_{0}^{\otimes N},
\]
 where $u$ has been identified with the following function on the
product $X^{N}:$ 

\[
u(x_{1},..,x_{N}):=\sum_{i=1}^{N}u(x_{i})
\]

\subsection{\label{sub:Preliminaries-on-Large dev}Preliminaries on Large Deviation
Principles and Legendre transforms}

Let us start by recalling the general definition of a Large Deviation
Principle (LDP) for a sequence of measures.
\begin{defn}
\label{def:large dev}Let $\mathcal{P}$ be a Polish space, i.e. a
complete separable metric space.

$(i)$ A function $I:\mathcal{\,P}\rightarrow]-\infty,\infty]$ is
a \emph{rate function} if it is lower semi-continuous. It is a \emph{good}
\emph{rate function} if it is also proper.

$(ii)$ A sequence $\Gamma_{k}$ of measures on $\mathcal{P}$ satisfies
a \emph{large deviation principle} with \emph{speed} $r_{k}$ and
\emph{rate function} $I$ if

\[
\limsup_{k\rightarrow\infty}\frac{1}{r_{k}}\log\Gamma_{k}(\mathcal{F})\leq-\inf_{\mu\in\mathcal{F}}I
\]
 for any closed subset $\mathcal{F}$ of $\mathcal{P}$ and 
\[
\liminf_{k\rightarrow\infty}\frac{1}{r_{k}}\log\Gamma_{k}(\mathcal{G})\geq-\inf_{\mu\in G}I(\mu)
\]
 for any open subset $\mathcal{G}$ of $\mathcal{P}.$ \end{defn}
\begin{rem}
The LDP is said to be \emph{weak }if the upper bound is only assumed
to hold when $\mathcal{F}$ is compact. Anyway, we will be mainly
interested in the case when $\mathcal{P}$ is compact and hence the
notion of a weak LDP and an LDP then coincide (and moreover any rate
functional is automatically good).
\end{rem}
We will be mainly interested in the case when $\Gamma_{k}$ is a probability
measure (which implies that $I\geq0$ with infimum equal to $0).$
Then it will be convenient to use the following alternative formulation
of a LDP (see Theorems 4.1.11 and 4.1.18 in \cite{d-z}): 
\begin{prop}
\label{prop:d-z}$\mathcal{P}$ be a metric  space and denote by $B_{\epsilon}(\mu)$
the ball of radius $\epsilon$ centered at $\mu\in\mathcal{P}.$ Then
a sequence $\Gamma_{N}$ of probability measures on $\mathcal{P}$
satisfies a weak LDP with speed $r_{N}$ and a rate functional $I$
iff 
\begin{equation}
\lim_{\epsilon\rightarrow0}\liminf_{N\rightarrow\infty}\frac{1}{r_{N}}\log\Gamma_{N}(B_{\epsilon}(\mu))=-I(\mu)=\lim_{\epsilon\rightarrow0}\limsup_{N\rightarrow\infty}\frac{1}{r_{N}}\log\Gamma_{N}(B_{\epsilon}(\mu))\label{eq:as on small balls}
\end{equation}

\end{prop}
We note the following simple lemma which allows one to extend the
previous proposition to the non-normalized measures $(\delta_{N})_{*}e^{-\beta H^{(N)}}\mu_{0}^{\otimes N}:$
\begin{lem}
\label{lem:ldp for nonnormalized meas}Assume that the following bound
for the partition functions holds: $\left|\log Z_{N,\beta}\right|\leq CN.$
Then the measures 
\begin{equation}
\Gamma_{N}:=(\delta_{N})_{*}e^{-\beta H^{(N)}}\mu_{0}^{\otimes N}\label{eq:def of Gamma N}
\end{equation}
 satisfy the asymptotics \ref{eq:as on small balls} for any $\mu\in\mathcal{M}_{1}(X)$
with rate functional $\tilde{I}(\mu)$ and $r_{N}=N$ iff the probability
measures $(\delta_{N})_{*}\mu_{\beta}^{(N)}$ on $\mathcal{M}_{1}(X)$
satisfy an LDP at speed $N$ with rate functional $I:=\tilde{I}-C_{\beta},$
where $C_{\beta}:=\inf_{\mathcal{\mu\in}\mathcal{M}(X)}I(\mu).$\end{lem}
\begin{proof}
Set $\tilde{\Gamma}_{N}:=(\delta_{N})_{*}e^{-\beta H^{(N)}}\mu_{0}^{\otimes N}$
and $C_{N,\beta}:=-\frac{1}{N}\log Z_{N,\beta}.$ By assumption $C_{N,\beta}$
is uniformly bounded and we denote by $C_{\beta}$ a given limit point
of the sequence obtained by replacing $N$ with a subsequence $N_{j}.$
Since $\frac{1}{N}\log\Gamma_{N}(B_{\epsilon}(\nu))=\frac{1}{N}\log\tilde{\Gamma}_{N}(B_{\epsilon}(\nu))+C_{N,\beta}$
we obtain that after replacing $N$ with the subsequence $N_{j}$
the probability measures $\Gamma_{N}$ satisfy (by Prop \ref{prop:d-z})
an LDP with rate functional $\tilde{I}-C_{\beta}.$ As a consequence
$0=\inf(\tilde{I}-C_{\beta}),$ showing that $C_{\beta}$ is independent
of the subsequence. Hence, the whole sequence converges towards $C_{\beta},$
which proves one direction in the Lemma. The converse is proved in
a similar way. 
\end{proof}
We will also use the following classical result of Sanov, which is
the standard example of a LDP for point processes \cite{d-z} (the
result follow, for example, from the Gärtner-Ellis theorem; see Section
\ref{sub:Relations-to-the g-e}).
\begin{prop}
\label{prop:sanov}Let $X$ be a topological space and $\mu_{0}$
a finite measure on $X.$ Then the law $\Gamma_{N}$ of the empirical
measures of the corresponding Gibbs measure $\mu_{0}^{\otimes N}$
(i.e. $H^{(N)}=0)$ satisfies an LDP with speed $N$ and rate functional
the relative entropy $D_{\mu_{0}}.$ 
\end{prop}
We recall that the \emph{relative entropy} $D_{\mu_{0}}$ (also called
the \emph{Kullback\textendash Leibler divergence }or the\emph{ information
divergence} in probability and information theory) is the functional
on $\mathcal{M}_{1}(X)$ defined by 
\begin{equation}
D_{\mu_{0}}(\mu):=\int_{X}\log\frac{\mu}{\mu_{0}}\mu,\label{eq:def of rel entropy}
\end{equation}
 when $\mu$ has a density $\frac{\mu}{\mu_{0}}$ with respect to
$\mu_{0}$ and otherwise $D_{\mu_{0}}(\mu):=\infty.$ When $\mu_{0}$
is a probability measure,$D_{\mu_{0}}(\mu)\geq0$ and $D_{\mu_{0}}(\mu)=0$
iff $\mu=\mu_{0}$ (by Jensen's inequality).

\subsubsection{Legendre-Fenchel transforms}

Let $f$ be a function on a topological vector space $V.$ Then its
The Legendre-Fenchel transform is defined as following convex lower
semi-continuous function $f^{*}$ on the topological dual $V^{*}$
\[
f^{*}(w):=\sup_{v\in V}\left\langle v,w\right\rangle -f(v)
\]
in terms of the canonical pairing between $V$ and $V^{*}.$ In the
present setting we will take $V=C^{0}(X)$ and $V^{*}=\mathcal{M}(X),$
the space of all signed Borel measures on a compact topological space
$X.$ We will use the following variant of the Brøndsted-Rockafellar
property $A^{*}$\cite{b-r}: 
\begin{lem}
\label{lem:leg-fench approx}Let $f$ be function on $C^{0}(X)$ which
is Gateaux differentiable. Then, for any $\mu\in\mathcal{M}(X)$ such
that $f^{*}(\mu)<\infty$ there exists a sequence of $u_{j}\in C^{0}(X)$
such that 
\begin{equation}
\mu_{j}:=df_{|u_{j}}\rightarrow\mu,\,\,\,f^{*}(\mu_{j})\rightarrow f^{*}(\mu)\label{eq:density property}
\end{equation}
\end{lem}
\begin{proof}
First recall that a convex function $g$ on a topological vector space
$E$ is said to be \emph{subdifferentiable} at $x\in E$ if $g(x)<\infty$
and $g$ admits a \emph{subgradient} $x^{*}$ at $x,$ i.e. an element
$x^{*}$in the topological dual $E^{*}$ such that for any $y\in E$

\[
g(y)\geq g(x)+\left\langle (y-x),x^{*}\right\rangle 
\]
The set of all such subgradients is denoted by $(\partial g)(x).$
Now assume that $g=f^{*}$ for a convex function $f$ on a Banach
space $V.$ Then $g$ is a lower semi-continuous function convex function
on the topological vector space $E:=V^{*}$ equipped with its weak
topology. According to \cite[Thm 2]{b-r} any element $\mu\in E^{*}$
such that $f^{*}(\mu)<\infty$ has the property that there exists
a sequence $\mu_{j}\rightarrow\mu$ in $V^{*}$ such that $f^{*}$
is subdifferentiable at $\mu_{j}$ with a subgradient in $E.$ Equivalently,
this means that there exists $u_{j}\in E$ such that $\mu_{j}\in(\partial f)(u_{j})$
(as follows from the definition the Legendre-Fenchel transform). Finally,
setting $E:=C^{0}(X)$ and observing that if $f$ is Gateaux differentiable
at $u\in V,$ then $(\partial f)(u)=\{df_{|u}\}$ (as is seen by restricting
$f$ to any affine line) thus concludes the proof.\end{proof}
\begin{rem}
By convexity, if $\mu=df_{|u}$ for some $u\in V:=C^{0}(X),$ then
$f^{*}(\mu)=\left\langle u,df_{|u}\right\rangle -f(u),$ which is
essentially the classical definition of the\emph{ Legendre transform}
of $f$ at $\mu.$ Accordingly, the previous lemma may be reformulated
as the statement that the Legendre-Fenchel transform is the greatest
lower semi-continuous extension to all of $V$ of the Legendre transform,
originally defined on $(df)(V)\subset V^{*}.$
\end{rem}

\subsection{The proof of Theorem \ref{thm:def determ}}

We start with the following simple
\begin{lem}
\label{lem:limit of Z as inf}Assume that $H^{(N)}$ satisfies the
quasi-superharmonicity assumption in the second point of Theorem \ref{thm:gibbs intro}\textup{.
Then, for any sequence of positive numbers $\beta_{N}\rightarrow\infty$
\[
-\mathcal{F}_{\beta_{N}}(u):=\frac{1}{N\beta_{N}}\log\int_{X^{N}}e^{-\beta_{N}(H^{(N)}+u)}dV^{\otimes N}=-\inf_{X^{N}}\frac{H^{(N)}+u}{N}+o(1)
\]
}\end{lem}
\begin{proof}
The inequality $\leq$ is trivial and to prove the reversed inequality
we fix a sequence of $x^{(N)}\in X^{N}$ realizing the infimum appearing
the right hand side above. Then replacing the integral of $X^{N}$
with an integral over the $L^{\infty}-$ball $B_{\epsilon}:=\{(x_{1},...,x_{N}):\,d_{g}(x,x_{i}^{(N)})\leq\epsilon\}^{N},$
for a fixed number $\mbox{\ensuremath{\epsilon}\ }$and a fixed metric
$g$ with distance function $d_{g},$ and using the classical submean
inequality in each variable with a fixed multiplicative constant $C$
gives 
\[
\int_{X^{N}}e^{-\beta_{N}(H^{(N)}+u)}dV^{\otimes N}\geq C^{N}e^{-\beta_{N}(H^{(N)}+u)(x^{(N)})}\int_{B_{\epsilon}}dV^{\otimes N}e^{-N\beta_{N}\delta_{\epsilon}}
\]
$\delta_{\epsilon}$ is the modulus of continuity of $u,$ tending
to $0$ as $\epsilon\rightarrow0.$ Finally, since $\int_{B_{\epsilon}}dV^{\otimes N}\geq(C'\epsilon)^{N}$
letting $N\rightarrow\infty$ concludes the proof. 
\end{proof}
To handle the case when $\beta_{N}=\beta+o(1)$ for a finite $\beta$
we will need to use the subexponential dependence on the dimensions
of the multiplicative constant appearing in Theorem \ref{thm:submean ineq text}.
To this end we first recall that, since $X$ is assumed compact, the
weak topology on $\mathcal{M}_{1}(X)$ is metrized by the Wasserstein
2-metric $d$ induced by a given Riemannian metric $g$ on $X,$ where
\[
d(\mu,\nu)^{2}:=\inf_{\Gamma\in\Gamma(\mu,\nu)}\int d_{g}(x,y)^{2}d\Gamma,
\]
where $\Gamma(\mu,\nu)$ is the space of all\emph{ couplings} between
$\mu$ and $\nu,$ i.e. all probability measures $\Gamma$ on $X\times X$
such that the push forward of $\Gamma$ to the first and second factor
is equal to $\mu$ and $\nu$ respectively. 
\begin{prop}
\label{prop:submean ineq hamilt}For any given $\epsilon>0$ there
exists a positive constant $C_{\epsilon}$ such that the following
submean inequality holds on $X^{N},$ for any $N:$ 
\begin{equation}
e^{-\beta H^{(N)})}(x^{(N)})\leq C_{\epsilon}e^{\epsilon N}\frac{\int_{B_{\epsilon}(x^{(N)})}e^{-\beta H^{(N)}}dV^{\otimes N}}{\int_{B_{\epsilon^{2}}(x^{(N)})}dV^{\otimes N}},\label{eq:submean proof thm ldp}
\end{equation}
where $B_{r}(x^{(N)})$ denotes the inverse image in $X^{N},$ under
the map $\delta_{N},$ of the Wasserstein ball of radius $r$ centered
at $\delta_{N}(x^{(N)})$\end{prop}
\begin{proof}
First observe that the pull-back of $d$ on $\mathcal{M}_{1}(X)$
to the quotient space $X^{(N)}:=X^{N}/S_{N}$ under the map $\delta_{N}$
defined by the empirical measure (formula \ref{eq:empirical measure text})
coincides with $1/N^{1/2}$ times the quotient distance function on
$X^{(N)},$ induced by the product Riemannian metric on $X^{N}:$
\begin{equation}
\delta_{N}^{*}d=\frac{1}{N^{1/2}}d_{X^{(N)}}:=d^{(N)}\label{eq:isometry prop in proof}
\end{equation}
Indeed, this is well-known and follows from the Birkhoff-Von Neumann
theorem which gives that for any symmetric function $c(x,y)$ on $X\times X$
we have that if $\mu=\frac{1}{N}\sum_{i=1}^{N}\delta_{x_{i}}$ and
$\nu=\frac{1}{N}\sum_{i=1}^{N}\delta_{y_{i}}$ for given $(x_{1},...,x_{N}),(y_{1},...,y_{N})\in X^{N},$
then 
\[
\inf_{\Gamma(\mu,\nu)}\int c(x,y)d\Gamma=\inf_{\Gamma_{N}(\mu,\nu)}\int c(x,y)d\Gamma
\]
where $\Gamma_{N}(\mu,\nu)\subset\Gamma(\mu,\nu)$ consists of couplings
of the form $\Gamma_{\sigma}:=\frac{1}{N}\sum\delta_{x_{i}}\otimes\delta_{y_{\sigma(i)}},$
for $\sigma\in S_{N},$ where $S_{N}$ is the symmetric group on $N$
letters. Now consider the metric space $(X^{(N)},d^{(N)})$ which
is the quotient space defined with respect to the finite group $S_{N}$
acting isometrically on the Riemannian manifold $(X^{N},g_{N}),$
where $g_{N}$ denotes $1/N$ times the product Riemannian metric.
By assumption $H^{(N)}$ is $S_{N}-$invariant and $\Delta_{g_{N}}H^{(N)}\leq C$
on $X^{N}$ (using the obvious scaling of the Laplacian). Moreover,
 since $X$ is compact there exists a non-negative number $k$ such
that $\mbox{Ric \ensuremath{g\geq-kg}}$ on $X$ and hence rescaling
gives $\mbox{Ric \ensuremath{g_{N}\geq-kNg_{N}}}$ on $(X^{N},g_{N}).$
But the dimension of $X^{N}$ is equal to $nN$ and hence setting
$\kappa^{2}:=k/n+1$ shows that, for $N$ large, the assumptions in
Theorem \ref{thm:submean ineq text} are satisfied for $u:=e^{-\beta H^{(N)}}$
and $(X,g)$ replaced by $(X^{N},g_{N}).$ Applying the latter theorem
with $\delta=\epsilon$ and using the pull-back property in formula
\ref{eq:isometry prop in proof} then shows that the submean property
\ref{eq:submean proof thm ldp} indeed holds.
\end{proof}
We will also rely on the following simple but very useful lemma (which
was used in the similar context of Fekete points in \cite{b-b-w}).
\begin{lem}
\label{lem:conv of abstr fekete}Fix $u_{*}\in C^{0}(X)$ and assume
that $x_{*}^{(N)}\in X^{N}$ is a minimizer of the function $(H^{(N)}+u_{*})/N$
on $X^{N}.$ If the corresponding large $N-$ limit $\mathcal{F}(u)$
exists for all $u\in C^{0}(X)$ and \emph{$\mathcal{F}$ is Gateaux
differentiable at $u_{*},$} then $\delta_{N}(x_{*}^{(N)})$ converges
weakly towards $\mu_{*}:=d\mathcal{F}_{|u_{*}}.$\end{lem}
\begin{proof}
Fix $v\in C^{0}(X)$ and a real number $t.$ Let $f_{N}(t):=\frac{1}{N}(H^{(N)}+u+tv)(x_{*}^{(N)})$
and $f(t):=\mathcal{F}(u+tv).$ By assumption $\lim_{N\rightarrow}f_{N}(0)=f(0)$
and $\liminf_{N\rightarrow}f_{N}(t)\geq f(t).$ Note that $f$ is
a concave function in $t$ (since it is defined as an inf of affine
functions) and $f_{N}(t)$ is affine in $t.$ But then it follows
from the differentiability of $f$ at $t=0$ that $\lim_{N\rightarrow\infty}df_{N}(t)/dt_{|t=0}=df(t)/dt_{|t=0},$
i.e. that 
\[
\lim_{N\rightarrow\infty}\left\langle \delta_{N}(x_{*}^{(N)}),v\right\rangle =\left\langle d\mathcal{F}_{|u},v\right\rangle ,
\]
 which thus concludes the proof of the lemma. 
\end{proof}

\subsection*{The upper bound in the LDP}

By Lemma \ref{lem:ldp for nonnormalized meas} it will be enough to
establish the LDP for the non-normalized measures $\Gamma_{N}$ in
formula \ref{eq:def of Gamma N}. To prove the upper bound of the
integrals appearing in the equivalent formulation of the LDP in Prop
\ref{prop:d-z} we fix a function $u\in C^{0}(X)$ and rewrite 
\[
e^{-\beta H^{(N)})}=e^{-\beta(H^{(N)}+u)}e^{\beta u},
\]
Then, trivially, for any fixed $\epsilon>0,$ 
\begin{equation}
\int_{B_{\epsilon}(\mu)}e^{-\beta H^{(N)}}dV^{\otimes N_{k}}\leq\sup_{B_{\epsilon}(\mu)}\left(e^{-\beta(H^{(N)}+u)}e^{\beta u}\right)\int_{X^{N}}\mu_{u}^{\otimes N},\,\,\mu_{u}:=e^{\beta u}dV\label{eq:trivial ineq in outline}
\end{equation}
Hence, replacing the sup over $B_{\epsilon}(\mu)$ with the sup over
all of $X^{N_{k}}$and applying Sanov's theorem relative to the tilted
volume form $\mu_{u}$ gives 
\[
\lim_{\epsilon\rightarrow0}\limsup_{k\rightarrow\infty}\frac{1}{\beta N}\log\int_{B_{\epsilon}(\mu)}e^{-\beta H^{(N)}}dV^{\otimes N}\leq-\mathcal{F}(u)+\int u\mu-\frac{1}{\beta}D_{dV}(\mu),
\]
 using that $D_{e^{\beta u}dV}(\mu)=-\beta\int u\mu+D_{dV}(\mu).$
Finally, taking the infimum over all $u\in C^{0}(X)$ shows that the
$\limsup$ in the previous formula is bounded from above by $-F(\mu),$
\[
F(\mu):=f^{*}(\mu)+\frac{1}{\beta}D_{dV}(\mu),\,\,\,\,f(u):=-\mathcal{F}(-u)
\]

\begin{rem}
In the argument above $dV$ can be replaced by any finite measure
$\mu_{0}$ on $X.$
\end{rem}

\subsection*{The lower bound in the LDP}

As usually the proof of the lower bound in the LDP is the hardest.
We first assume that 
\[
\mu=d\mathcal{F}_{|u}
\]
 for some $u\in C^{0}(X).$ Denote by $x^{(N)}\in X^{N}$ a sequence
of minimizers of $H^{(N)}+u.$ By Lemma \ref{lem:conv of abstr fekete}
we have that
\[
\delta(x^{(N)})\rightarrow\mu
\]
 weakly and hence for $N$ sufficiently large $\Gamma_{N}(B_{2\epsilon}(\mu)):=\int_{B_{2\epsilon}(\mu)}e^{-\beta H^{(N)}}dV^{\otimes N}\geq$

\emph{
\[
\geq\int_{B_{\epsilon}(x^{(N)})}e^{-\beta(H^{(N)}+u)}e^{\beta u}dV^{\otimes N}\geq e^{N\beta\left\langle u,\nu\right\rangle -N\delta(\epsilon)}\int_{B_{\epsilon}(x^{(N)})}e^{-\beta(H^{(N)}+u)}dV^{\otimes N}
\]
}where $\delta(\epsilon)$ is a modulus of continuity for $u$ on
$X$ tending to zero with $\epsilon$ (by the compactness of $X).$
Next, applying the submean inequality \ref{eq:submean proof thm ldp}
gives 
\[
\frac{1}{N}\log\Gamma_{N}(B_{2\epsilon}(\mu))\geq\beta\left\langle u,\mu\right\rangle -\delta(\epsilon)+\frac{1}{N}\log\int_{B_{\epsilon^{2}}(x^{(N)})}dV^{\otimes N}+\beta(H^{(N)}+u)(x^{(N)})/N-\epsilon-\frac{C_{\epsilon}}{N}
\]
Since $\delta(x^{(N})\rightarrow\mu$ we may, for $N$ sufficiently
large, assume that $B_{\epsilon^{2}/2}(\mu)\subset B_{\epsilon^{2}}(\delta(x^{(N)}))$
and hence letting $N\rightarrow\infty$ and using Sanov's theorem
(i.e. Prop \ref{prop:sanov}) for $\epsilon$ fixed and the assumed
convergence of $(H^{(N)}+u)(x^{(N)})/N$ gives 
\[
\liminf_{N\rightarrow\infty}\frac{1}{N}\log\Gamma_{N}(B_{2\epsilon}(\mu))\geq\beta\left\langle u,\mu\right\rangle -\delta(\epsilon)-\inf_{B_{\epsilon^{2}/2}}D_{dV}+\beta\mathcal{F}(u)-\epsilon
\]
Since $\mu$ is a candidate for the inf in the right hand side above
the inf in question may be estimated from above by $D_{dV}(\mu)$
and hence letting $\epsilon\rightarrow0$ concludes the proof under
the assumption that $\mu:=d\mathcal{F}_{|u}$ for some $u\in C^{0}(X).$
To prove the general case we invoke Lemma \ref{lem:leg-fench approx}
to write $\mu$ as a weak limit of $\mu_{j}:=d\mathcal{F}_{|u_{j}}$
for $u_{j}\in C^{0}(X).$ We may then replace $u$ in the previous
argument with $u_{j}$ for a fixed $j$ and replace $\mu$ with $\mu_{j}$
in the previous argument to get, for $j\geq j_{\delta},$ $\liminf_{N\rightarrow\infty}\frac{1}{N}\log\Gamma_{N}(B_{3\epsilon}(\mu))\geq$
\[
\geq\liminf_{N\rightarrow\infty}\frac{1}{N}\log\Gamma_{N}(B_{2\epsilon}(\mu_{j}))\geq\beta\left\langle u_{j},\mu_{j}\right\rangle -\delta_{j}(\epsilon)-\inf_{B_{\epsilon^{2}/2}(\nu_{j})}D_{dV}+\beta\mathcal{F}(u_{j})-\epsilon
\]
But for $j$ sufficiently large $\mu_{j}$ is in the ball $B_{\epsilon^{2}/2}(\nu)$
and hence the inf above is bounded from above by $D_{dV}(\mu)$ giving
\[
\liminf_{N\rightarrow\infty}\frac{1}{N}\log\Gamma_{N}(B_{3\epsilon}(\mu))\geq\beta\left\langle u_{j},\mu_{j}\right\rangle -\delta_{j}(\epsilon)-D_{dV}(\mu)+\beta\mathcal{F}(u_{j})-\epsilon
\]
Letting first $\epsilon\rightarrow0$ and then $j\rightarrow\infty$
gives 
\[
\liminf_{N\rightarrow\infty}\frac{1}{N}\log\Gamma_{N}(B_{3\epsilon}(\mu))\geq-\beta(\lim_{j\rightarrow\infty}(E(\mu_{j})+\frac{1}{\beta}D_{dV}(\mu))
\]
Finally, by Lemma \ref{lem:leg-fench approx} we may assume that $E(\mu_{j})\rightarrow E(\mu)$
and that concludes the proof.

\subsection*{The equation for the minimizer $\mu_{\beta}$}

Finally, the equation \ref{eq:abstract mean field eq intro} follows
immediately from the following general convex analytical result: 
\begin{lem}
\label{lem:R-F duality etc}Let $X$ be a compact topological space
and $f$ and \emph{$g$ be Gateaux differentiable convex functionals
on $C^{0}(X)$ such that the differentials $dg$ and $df$ takes values
in $\mathcal{M}_{1}(X).$ Then }
\begin{itemize}
\item \emph{The following identity holds:
\begin{equation}
\inf_{\mathcal{M}_{1}(X)}\left(f^{*}+g^{*}\right)=\sup_{u\in C^{0}(X)}\left(-f(-u)-g(u)\right)\label{eq:inf is sup in lemma}
\end{equation}
}
\item \emph{if the sup in the right hand side above is attained at some
$u_{0}$ in $C^{0}(X)$ (i.e. if $-f(-u)-g(u)$ admits a critical
point $u_{0}),$ then, setting $\mathcal{F}(u):=-f(-u),$ the measure
$\mu_{0}:=d\mathcal{F}_{|u_{0}}$ minimizes the functional $f^{*}+g^{*}$
on $\mathcal{M}_{1}(X).$ }
\end{itemize}
\end{lem}
\begin{proof}
First observe that $f$ and $g$ are Lipschitz continuous on the Banach
space $C^{0}(X).$ Indeed, setting $u_{t}:=u_{0}(1-t)+tu_{1},$ for
$t\in[0,1],$ gives 
\[
\left|f(u)-f(v)\right|=\left|\int_{0}^{1}dt\int_{X}df_{u_{t}}(u_{1}-u_{0})\right|\leq\sup_{X}\left|u_{1}-u_{0}\right|
\]
and similarly for $g.$ The first point in the lemma is then obtained
as a special case of the Fenchel-Rockafeller duality theorem which
only requires that $f$ and $g$ be convex on a Banach space $V$
and that $f$ and $g$ be finite at some point $u$ where $f$ is
moreover assumed continuous \cite[Thm 1.12]{br}. To prove the second
point we let $u_{0}$ be a critical point of $\mathcal{F}(u)-g(u)$
on $C^{0}(X),$ i.e. 
\begin{equation}
d\mathcal{F}_{|u_{0}}=dg_{|u_{0}},\label{eq:crit pt eq}
\end{equation}
which, by convexity, means that $u_{0}$ realizes the sup in the right
hand side of formula \ref{eq:inf is sup in lemma}. We rewrite, 
\begin{equation}
f^{*}(\mu):=\sup_{u\in C^{0}(X)}\left\langle u,\mu\right\rangle -f(u)=\sup_{u\in C^{0}(X)}\mathcal{F}(u)-\left\langle u,\mu\right\rangle \label{eq:proof of lemma inf is sup}
\end{equation}
(by replacing $u$ with $-u$ in the sup). Hence, if $\mu:=d\mathcal{F}_{|u}$
then, by concavity, $f^{*}(\mu):=F(u)-\left\langle u,\mu\right\rangle .$
Similarly, if $\mu=dg_{|v}$ then, by convexity, $g^{*}(\mu):=\left\langle u,\mu\right\rangle -g(u).$
All in all this means that if $u_{0}$ satisfies the critical point
equation \ref{eq:crit pt eq}, then we can take $u=v=u_{0}$ to get
\[
f^{*}(\mu_{0})+g^{*}(\mu_{0})=\mathcal{F}(u_{0})+0-g(u_{0}),
\]
 which concludes the proof, using the first point.
\end{proof}

\section{\label{sec:Relations-to-convergence,}Relations to $\Gamma-$convergence,
the Gärtner-Ellis theorem and mean energy}

Before turning to the applications of Theorem \ref{thm:gibbs intro}
in the complex geometric setting we explore some relations to previous
results and methods in the literature.

\subsection{\label{sub:Relations-to-Gamma}Relations to Gamma convergence }

We recall that a sequence of functions $E_{N}$ on a topological space
\emph{$\mathcal{P}$ is said to $\Gamma-$converge }to a function
$E$ on $\mathcal{P}$ if 
\begin{equation}
\begin{array}{ccc}
\mu_{N}\rightarrow\mu\,\mbox{in\,}\mathcal{P} & \implies & \liminf_{N\rightarrow\infty}E_{N}(\mu_{N})\geq E(\mu)\\
\forall\mu & \exists\mu_{N}\rightarrow\mu\,\mbox{in\,}\mathcal{P}: & \lim_{N\rightarrow\infty}E_{N}(\mu_{N})=E(\mu)
\end{array}\label{eq:def of gamma conv}
\end{equation}
(such a sequence $\mu_{N}$ is called a\emph{ recovery sequence});
see \cite{bra}. It then follows that $E$ is lower semi-continuous
on $\mathcal{P}.$ In the present setting we take, as before, $\mathcal{P}=\mathcal{M}(X)$
and define $E_{N}$ by setting $E_{N}=\infty$ on the complement of
the image of the map$\delta_{N}$ and 
\begin{equation}
E_{N}(\delta_{N}(x_{1},...,z_{N}):=H^{(N)}(x_{1},...,x_{N})/N\label{eq:energi on M_1 induced by Hamilt}
\end{equation}
We can now formulate the following variant of Theorem \ref{thm:gibbs intro}: 
\begin{thm}
\label{thm:g-conv of sh energies}Let $H^{(N)}$ be a sequence of
lower semi-continuous symmetric functions on $X^{N},$ where $X$
is a compact Riemannian manifold.\textup{ Assume that }
\begin{itemize}
\item \textup{The functions $E_{N}$ on $\mathcal{M}_{1}(X)$ determined
by $H^{(N)}$ converge to a function $E,$ in the sense of $\Gamma$
convergence on $\mathcal{M}_{1}(X).$}
\item $H^{(N)}$ is uniformly quasi-superharmonic, i.e. \emph{$\Delta_{x_{1}}H^{(N)}(x_{1},x_{2},...x_{N})\leq C$
on $X^{N}$}
\end{itemize}
Then, for any sequence of positive numbers $\beta_{N}\rightarrow\beta\in]0,\infty]$
the measures $\Gamma_{N}:=(\delta_{N})_{*}e^{-\beta_{N}H^{(N)}}$
on $\mathcal{M}_{1}(X)$ satisfy, as $N\rightarrow\infty,$ a LDP
with \emph{speed} $\beta_{N}N$ and good \emph{rate functional} 
\begin{equation}
F_{\beta}(\mu)=E(\mu)+\frac{1}{\beta}D_{dV}(\mu)\label{eq:free energy func theorem gibbs intro-1}
\end{equation}
\end{thm}
\begin{proof}
Using the characterization of a LDP in Proposition \ref{prop:d-z},
the upper bound in the LDP follows almost immediately from the liminf
property of the Gamma-convergence together with Sanov's theorem. To
prove the lower bound fix $\mu\in\mathcal{M}_{1}(X)$ and take a recovery
sequence $\mu_{N}$ corresponding to a sequence $x^{(N)}\in X^{N}.$
Then, using the same notation for the balls as in the proof of Theorem
\ref{thm:gibbs intro}, we have, for $\epsilon>0$ fixed and $N$
large\emph{
\[
\int_{B_{2\epsilon}(\mu)}e^{-\beta H^{(N)}}dV^{\otimes N}\geq\int_{B_{\epsilon}(x^{(N)})}e^{-\beta H^{(N)}}dV^{\otimes N}\geq e^{NC\epsilon}e^{-N\beta E_{N}(\mu_{N})}\int_{B_{\epsilon^{2}}(x^{(N)})}dV^{\otimes N},
\]
}using the submean inequality in Theorem \ref{thm:submean ineq text}
in the last inequality. Letting first $N\rightarrow\infty$ and then
$\epsilon\rightarrow0$ then concludes the proof, using Sanov's theorem
again.
\end{proof}
It should be stressed that, in general, the functional $E(\mu)$ in
the previous theorem will not be convex and hence the subset $\mathcal{C}_{\beta}\subset\mathcal{M}_{1}(X)$
consisting of the minima of $F_{\beta}$ will, in general, consist
of more than one element. By general principles the LDP then implies
that any limit point $\Gamma_{\infty}\in\mathcal{M}_{1}\left(\mathcal{M}_{1}(X)\right)$
of the laws $\Gamma_{N}$ is concentrated on $\mathcal{C}_{\beta}$
(in the terminology of statistical mechanics $\Gamma_{\infty}$ is
thus a \emph{mixed state} defined as a superposition of the pure states
$\delta_{\mu}$ where $\mu\in\mathcal{C}_{\beta}).$ 
\begin{rem}
The proof of the previous theorem in the case $\beta=\infty$ is much
simpler as it is does not require the sub-exponential dependence on
the dimension in the submean inequality in Theorem \ref{thm:submean ineq text}.
Indeed, the rough exponential bound used in in the proof of Lemma
\ref{lem:limit of Z as inf} is enough. Moreover, all that is used
in the proof for $\beta<\infty$ is that $\Delta_{x_{1}}(e^{-\beta_{N}H^{(N)}})\geq-\lambda_{\beta}e^{-\beta_{N}H^{(N)}}$
for a constant $\lambda_{\beta}$ independent on $N$ (but the assumption
that $\Delta_{x_{1}}H^{(N)}\leq C$ is a convenient way of ensuring
that the previous inequality holds for any $\beta).$\end{rem}
\begin{example}
\label{exa:gamma conv of mean field}In the case when $X=\R^{n}$
equipped with the Euclidean distance it is known that the mean field
Hamiltonian with pair interaction of the form $W(x,y)=w(|x-y|)$ (formula
\ref{eq:mean field pair hamilton}) $\Gamma-$convergences towards
$E(\mu):=\int_{X^{2}}W\mu\otimes\mu,$ if $w$ is lower semi-continuous
and increasing close to $0$ (see \cite[Prop 2.8, Remark 2.19]{se}
and \cite{ben-g,ber-o,c-g-z} for similar results). The proof exploits
the explicit nature of $E(\mu)$ and a similar argument applies on
a compact manifold when $W$ is continuous away from the diagonal
with a singularity of the local form $w(|x-y|)$ close to the diagonal
(compare \cite{z-z,z2}).
\end{example}
In contrast to the previous example, for the ``determinantal'' Hamiltonian
\ref{eq:def of Hamtiltonian in complex setting text} appearing in
the complex geometric setting there is no explicit candidate for a
limit $E(\mu).$ Instead the Gamma convergence is a consequence of
the following dual criterion.

\subsubsection{A criterion for Gamma convergence using duality }

Next we separate out the convex analysis used in the proof of Theorem\ref{thm:gibbs intro}
to get the following criterion for $\Gamma-$convergence:
\begin{prop}
Let $E_{N}$ a sequence of functions on $\mathcal{M}_{1}(X)$ and
assume that
\[
\lim_{N\rightarrow\infty}E_{N}^{*}(u)=f(u)
\]
 where $f$ is a Gateaux differentiable convex function on $C^{0}(X).$
Then $E_{N}$ converges to $E:=f^{*}$in the sense of $\Gamma-$convergence
on the space $\mathcal{M}_{1}(X),$ equipped with the weak topology. \end{prop}
\begin{proof}
First suppose that $\mu_{N}\rightarrow\mu$ weakly in $\mathcal{M}_{1}(X).$
Fix $u$ in $C^{0}(X).$ Then $-E_{N}(\mu_{N})=\left\langle u,\mu_{N}\right\rangle -E_{N}(\mu_{N})-\left\langle u,\mu\right\rangle +o(1)$
and hence taking the sup over all $\mu\in\mathcal{M}_{1}(X)$ gives
\[
-E_{N}(\mu_{N})\leq f_{N}(u)-\left\langle u,\mu\right\rangle +o(1)=f(u)-\left\langle u,\mu\right\rangle +o(1).
\]
 Finally, letting first $N\rightarrow\infty$ and then taking the
sup over all $u\in C^{0}(X)$ concludes the proof of the lower bound
for $E_{N}(\mu_{N}).$

To prove the existence of a recovery sequence we first assume that
$\mu=df_{|u_{\mu}}$for some $u_{\mu}\in C^{0}(X).$ Then, by
\[
f^{*}(\mu)=\left\langle u_{\mu},\mu\right\rangle -f(u_{\mu})=\left\langle u_{\mu},\mu\right\rangle -f_{N}(u_{\mu})+o(1),
\]
 since, by assumption, $f_{N}(u)=f(u)+o(1).$ Now, by the weak compactness
of $\mathcal{M}_{1}(X)$ the sup defining $f_{N}$ is attained at
some $\mu_{N}\in\mathcal{M}_{1}(X)$ and hence 
\[
f^{*}(\mu)+o(1)=\left\langle u_{\mu},\mu\right\rangle -\left(\left\langle u_{\mu},\mu_{N}\right\rangle -E_{N}(\mu_{N})\right)
\]
Next, by a minor generalization of Lemma \ref{lem:conv of abstr fekete}
$\mu_{N}\rightarrow\mu(:=df_{|u_{\mu}})$ and hence $f^{*}(\mu)=0+E_{N}(\mu_{N})+o(1),$
as desired. Finally, the proof of the existence of recovery sequence
for any $\mu$ such that $E(\mu)<\infty$ is concluded by a simple
 diagonal argument based on Lemma \ref{lem:leg-fench approx} applied
to $E:=f^{*}.$ 
\end{proof}
Now, if $E_{N}$ is of the form \ref{eq:energi on M_1 induced by Hamilt},
then 
\begin{equation}
f_{N}(u):=\sup_{X^{N}}\frac{1}{N}u(x_{1})+...+\frac{1}{N}u(x_{N})-\frac{1}{N}H^{(N)}(x_{1},...,x_{N})\label{eq:def of f_N in terms of H}
\end{equation}
Thanks to the previous proposition the first assumption in Theorem
\ref{thm:gibbs intro} thus implies (also using Lemma \ref{lem:limit of Z as inf})
that $E_{N}\rightarrow E$ in the sense of $\Gamma-$convergence on
$\mathcal{M}(X).$ Accordingly we recover Theorem \ref{thm:gibbs intro}
from Theorem \ref{thm:g-conv of sh energies}. 
\begin{rem}
\label{rem:duality gamma conv}In general, if $E_{N}$ gamma converges
to a function $E$ on $\mathcal{M}_{1}(X),$ then it follows (almost
directly) that $E_{N}^{*}\rightarrow E^{*}$ point-wise on $C^{0}(X).$
Hence, the point of the previous proposition is that it gives a converse
statement under the assumption that $E^{*}$ is Gateaux differentiable.
By basic convex duality it thus follows from the previous proposition
that $E_{N}$ converges to a strictly convex functional $E$ on $\mathcal{M}_{1}(X)$
iff $E_{N}^{*}\rightarrow E^{*}$ point-wise on $C^{0}(X),$ with
$E^{*}$ Gateaux differentiable.
\end{rem}

\subsection{\label{sub:Relations-to-the g-e}Relations to the Gärtner-Ellis theorem}

First observe that
\[
\int_{X^{N}}e^{-\beta_{N}(H^{(N)}+u)}dV^{\otimes N}=\widehat{\Gamma_{N}}(-r_{N}u),
\]
 where $\Gamma_{N}$ is the measure $(\delta_{N})_{*}(e^{-\beta_{N}H^{(N)}}dV^{\otimes N})$
on $\mathcal{M}_{1}(X)$ and $\widehat{\Gamma_{N}}$ denotes its Laplace
transform on $C^{0}(X).$ The Gärtner-Ellis theorem may, applied to
the sequence of measures $\Gamma_{N}$ on $\mathcal{M}_{1}(X),$ viewed
as a subset of the locally convex Hausdorff topological vector space$\mathcal{M}(X),$
may then be formulated as follows (see (see \cite[Cor 4.6.14, p. 148]{d-z}
and references therein):
\begin{thm}
\label{thm:g-e}Let $H^{(N)}$ be a sequence of Hamiltonians on $X^{N}$
and $\beta_{N}$ a sequence of positive numbers such that $\beta_{N}\rightarrow\beta\in]0,\infty].$
Assume that, for any $u\in C^{0}(X),$ as $N\rightarrow\infty,$ \textup{
\begin{equation}
\mathcal{F}_{\beta_{N}}(u):=-\frac{1}{N\beta_{N}}\log\int_{X^{N}}e^{-\beta_{N}(H^{(N)}+u)}dV^{\otimes N}\rightarrow\mathcal{F}_{\beta}(u),\label{eq:conv in g-e}
\end{equation}
 where $\mathcal{F}$ is a Gateaux differentiable function. }Then
the measures $\Gamma_{N}:=(\delta_{N})_{*}(e^{-\beta H^{(N)}}dV^{\otimes N})$
on $\mathcal{M}_{1}(X)$ satisfy, as $N\rightarrow\infty,$ an LDP
with \emph{speed} $\beta_{N}N$ and good \emph{rate functional $f^{*}(\mu),$
where $f(u):=-\mathcal{F}(-u).$}
\end{thm}
Compared with the Gärtner-Ellis theorem the main point of Theorem
\ref{thm:gibbs intro} is thus that, under the quasi-subharmonicity
assumption in the second point of the theorem, the assumption that
the convergence of the partition functions in formula \ref{eq:conv in g-e}
holds for $\beta=\infty$ is enough to ensure that one gets an LDP
for any $\beta<\infty.$ Moreover, as a consequence, the convergence
then also hold for any $\beta<\infty$ with $-\mathcal{F}_{\beta}(\cdot)$
defined as the Legendre-Fenchel transform of the rate functional $F_{\beta}$
appearing in Theorem \ref{thm:gibbs intro}. In fact, the latter convergence
is equivalent to the LDP in question, as made precise by the following 
\begin{lem}
\label{lem:appl of ge}Let $H^{(N)}$ be a sequence of Hamiltonians
on $X^{N}$ and $\beta_{N}$ a sequence of positive numbers such that
$\beta_{N}\rightarrow\beta\in]0,\infty[.$ Assume that, for any given
volume form $dV,$ the corresponding partition functions $Z_{N.\beta_{N}}$
satisfy
\[
\lim_{N\rightarrow\infty}-\frac{1}{N\beta_{N}}\log Z_{N.\beta_{N}}:=\inf_{\mu}F_{\beta},\,\,\,F_{\beta}:=E+D_{dV}/\beta,
\]
 with $E(\mu)$ convex. Then the measures $(\delta_{N})_{*}(e^{-\beta H^{(N)}}dV^{\otimes N})$
on $\mathcal{M}_{1}(X)$ satisfy, as $N\rightarrow\infty,$ an LDP
with \emph{speed} $\beta_{N}N$ and good \emph{rate functional $F_{\beta}.$
Moreover, if the asymptotics above also holds for $\beta=\infty$
with $E(\mu)$ strictly convex, then the LDP holds for $\beta=\infty,$
as well. }\end{lem}
\begin{proof}
Fixing a volume form $dV$ and applying the asymptotics in the lemma
to the volume forms $e^{-\beta u}dV$ for any $u\in C^{0}(X)$ reveals
that the asymptotics \ref{eq:conv in g-e} hold with $f_{\beta}$
given by the Legendre-Fenchel transform of $E+D_{dV}/\beta.$ Now,
if $E$ is convex, then $E+D_{dV}/\beta$ is strictly convex (since
$D_{dV}$ is) and hence it follows from basic convex duality that
$f_{\beta}$ is Gateaux differentiable. In fact, the differential
$\mu_{u}:=df_{\beta|u}$ is the unique minimizer attaining the sup
defining $f_{\beta}(u),$ viewed as the Legendre-Fenchel transform
of $E+D_{dV}/\beta.$ Equivalently, $\mu_{u}$ is the unique minimizer
of the strictly convex functional $\mu\mapsto E(\mu)+\left\langle u,\mu\right\rangle +D_{dV}(\mu)/\beta.$\end{proof}
\begin{rem}
\label{rem:g-e ldp conv}Let $\beta_{N}$ be sequence tending to $\infty.$
By convex duality the Gärtner-Ellis theorem may in the present setting,
be formulated as follows (also using Varadhan's lemma \cite{d-z}
in the converse): let $E_{N}$ be a sequence of functions on $\mathcal{M}_{1}(X).$
Then $e^{-\beta_{N}NE_{N}}(\delta_{N})_{*}(dV^{\otimes N})\sim e^{-\beta_{N}NE(\mu)}$
in the sense of a LDP, with $E(\mu)$ strictly convex iff $\beta_{N}N$
times the log of the Laplace transform of $e^{-\beta_{N}NE_{N}}(\delta_{N})_{*}(dV^{\otimes N})$
converges to the Gateaux differentiable function $E^{*}$on $C^{0}(X).$ 
\end{rem}

\subsection{\label{sub:Relations-to-the exist mean}Relations to the existence
of the mean energy}

Given a sequence of Hamiltonians $H^{(N)}$ on $X^{N}$ we set 
\[
\bar{E}_{N}(\mu):=\frac{1}{N}\int_{X^{N}}H^{(N)}\mu^{\otimes N},
\]
If the limit as $N\rightarrow\infty$ exists then we will call it
the\emph{ mean energy }of $\mu,$ denoted by $\bar{E}(\mu).$ 
\begin{example}
If $H^{(N)}$ is the mean field Hamiltonian associated to the pair
interaction potential $W$ (formula \emph{\ref{eq:mean field pair hamilton})}
then, trivially, $E(\mu)=\bar{E}_{N}(\mu)$ for any $\mu$ such that
$W\in L^{1}(\mu).$ 
\end{example}
It follows immediately from the definition that if the limit of $E_{N}^{*}(:=f_{N})$
appearing in formula \ref{eq:def of f_N in terms of H} exists then
\[
\bar{E}(\mu)\geq f^{*}(\mu).
\]
(but, in general this is a strict inequality, for example if $\bar{E}(\mu)$
is not convex). In particular, under the assumptions in Theorem \ref{thm:gibbs intro}
we have $\bar{E}(\mu)\geq E(\mu),$ where $E(\mu)$ appears as the
rate functional in Theorem \ref{thm:gibbs intro} for $\beta=\infty$
(using Lemma \ref{lem:limit of Z as inf}). Motivated by the complex
geometric applications discussed in Section \ref{sec:Outlook} this
leads one to consider the following
\begin{problem}
\label{prob:mean en}Show that the assumptions on $H^{(N)}$ in Theorem
\ref{thm:gibbs intro} imply that the corresponding mean energy $\bar{E}(\mu)$
exists when $\mu$ is a volume for (perhaps under additional appropriate
assumptions on $H^{(N)}$). 
\end{problem}
As illustrated by the following lemma this problem turns out to be
related to the asymptotics of the Gibbs measures with $\beta$ negative:
\begin{lem}
\label{lem:asympt for beta negative implies exist of mean}Assume
that there exists some negative $\beta_{0}$ such that for any $\beta\geq\beta_{0}$
the corresponding Gibbs measures are well-defined, i.e. $Z_{N.\beta_{N}}<\infty$
for $N$ sufficiently large. Moreover, assume that there exists a
functional $E(\mu)$ such that 
\begin{equation}
-\lim_{N\rightarrow\infty}\frac{1}{N}\log Z_{N.\beta_{N}}=\inf_{\mu\in\mathcal{M}_{1}(X)}\beta E(\mu)+D_{dV}(\mu)>-\infty,\label{eq:asympt in formula negative}
\end{equation}
 for any volume form $dV.$ Then the mean energy $\bar{E}(\mu)$ exists
for any volume form $\mu$ and $\bar{E}(\mu)=E(\mu).$\end{lem}
\begin{proof}
First observe that, by Jensen's inequality, the number $f_{N}(\beta):=-\frac{1}{N}\log Z_{N.\beta}$
appearing in the right hand side above for $N$ is concave in $\beta$
(and, by assumption, finite). Moreover, $\partial f_{N}(\beta)/\partial\beta=\bar{E}_{N}(dV)$
at $\beta=0.$ Further more, the finite function $f(\beta)$ defined
by the right hand side in formula \ref{eq:asympt in formula negative}
is also concave, as it is an infimum of a family of linear functions
and for $\beta=0$ the infimum attained is precisely at $\mu=dV.$
Hence, by basic convex analysis, the derivative of $f$ at $\beta=0$
exists and is given by $E(dV).$ Finally, the proof is concluded by
using that if $f_{N}$ is a sequence of convex functions converging
point-wise to convex function $f$ such that $f_{N}$ and $f$ are
differentiable at $0$ then the corresponding derivatives at $0$
also converge.
\end{proof}
It should, however, be stressed that, if $H^{(N)}$ is too singular
then the partition function $Z_{N.\beta}$ is equal to $\infty,$
for any $\beta<0$ (even if $H^{(N)}$ is quasi-superharmonic as in
the assumptions of Theorem \ref{thm:gibbs intro}). Indeed, for the
mean field Hamiltonian corresponding to a pair interaction $W$ this
happens as soon as $W$ has a repulsive power-law singularity, i.e.
$W(x,y)\sim|x-y|^{\alpha}$ with $\alpha<0$ close to the diagonal.
On the other hand, in the case of a \emph{logarithmic} singularity
$Z_{N.\beta_{N}}$ is indeed finite for $\beta_{0}<0$ and sufficiently
close to $0$ (see \cite{bo-g} for the corresponding LDP in the setting
of the 2D vortex model). 

Using the Gibbs variational principle some converses to Lemma \ref{lem:asympt for beta negative implies exist of mean}
can be established \cite{berm10}, where the existence of the mean
energy is assumed (and some additional assumptions), by extending
the approach of Messer-Spohn \cite{m-s}. However is should be stressed
that the main point of our proof of Theorem \ref{thm:gibbs intro}
is that it does note rely on the existence of the mean energy $\bar{E}(\mu),$
which, as pointed above, is an open problem in the present setting.

\section{\label{sec:Applications-to-K=0000E4hler-Einstein}Applications to
Kähler-Einstein geometry}

In this section we will apply Theorem \ref{thm:gibbs intro} to complex
manifolds $X$ equipped with a line bundle $L,$ assuming that $L$
is positive. The extension to big line bundles (and varieties of positive
dimension) is given in the companion paper\cite{berm8}, using the
full power of the pluripotential theory developed in \cite{begz,bbgz,berm6}
(see the discussion in Section \ref{sub:The-generalization-to big}).

\subsection{\label{sub:K=0000E4hler-geometry-setup}Kähler geometry setup}

Let $X$ be an $n-$dimensional compact complex manifold and denote
by $J$ the corresponding complex structure viewed as an endomorphism
of the real tangent bundle satisfying $J^{2}=-I.$

\subsubsection{Kähler forms/metrics}

On a complex manifold $(X,J)$ anti-symmetric two forms $\omega$
and symmetric two tensors $g$ on $TX\otimes TX,$ which are $J-$invariant,
may be identified by setting 
\[
g:=\omega(\cdot,J\cdot)
\]
Such a real two form $\omega$ is said to be Kähler if $d\omega=0$
and the corresponding symmetric tensor $g$ is positive definite (i.e.
defines a Riemannian metric)\footnote{In the complex analysis literature a $J-$invariant two form $\omega$
is usually said to be of type $(1,1)$ since $\omega=\sum_{i,j}\omega_{ij}dz_{i}\wedge d\bar{z}_{j}$
in local holomorphic coordinates.}. Conversely, a Riemannian metric $g$ is said to be \emph{Kähler
}if it arises in this way (in Riemannian terms this means that parallel
transport with respect to $g$ preserves $J).$ By the local $\partial\bar{\partial}-$
lemma a two form $\omega$ is closed, i.e. $d\omega=0$ if and only
if $\omega$ may be locally expressed as $\omega=\frac{i}{2\pi}\partial\bar{\partial}\phi,$
in terms of a local smooth function $\phi$ (called a local potential
for $\omega)$. In real notation this means that 
\[
\omega=dd^{c}\phi,\,\,d^{c}:=-\frac{1}{4\pi}d(J\cdot)
\]
(and hence $\omega$ is Kähler iff $\phi$ is strictly plurisubharmonic).
The normalization above ensures that $dd^{c}\log|z|^{2}$ is a probability
measure on $\C.$ We will denote by $[\omega]\in H^{2}(X,\R)$ the
de Rham cohomology represented by $\omega.$ If $\omega_{0}$ is a
fixed Kähler form then, according to the global $\partial\bar{\partial}-$
lemma, any other Kähler metric in $[\omega_{0}]$ may be globally
expressed as 
\[
\omega_{\varphi}:=\omega_{0}+dd^{c}\varphi,\,\,\,\,\varphi\in C^{\infty}(X),
\]
where $\varphi$ is determined by $\omega_{0}$ up to an additive
constant. We set 
\[
\mathcal{H}(X,\omega):=\left\{ \varphi\in C^{\infty}(X):\,\omega_{\varphi}>0\right\} 
\]
The association $\varphi\mapsto\omega_{\varphi}$ thus allows one
to identify $\mathcal{H}(X,\omega)/\R$ with the space of all Kähler
forms in $[\omega_{0}].$

\subsubsection{\label{sub:Metrics-on-line}Metrics on line bundles and curvature}

Let $L$ be a holomorphic line bundle on $X$ and $\left\Vert \cdot\right\Vert $
a Hermitian metric on $L.$ The normalized curvature two form of $\left\Vert \cdot\right\Vert $
may be (locally) written as 
\[
\omega:=-dd^{c}\log\left\Vert s\right\Vert ^{2},
\]
 in terms of a given local trivialization holomorphic section $s$
of $L.$ The corresponding cohomology class $[\omega]$ is independent
of the metric $\left\Vert \cdot\right\Vert $ on $L$ and coincides
with the first Chern class $c_{1}(L)$ in $H^{2}(X,\R)\cap H^{2}(X,\Z)$
(conversely, any such cohomology class is the first Chern class of
line bundle $L).$ A line bundle $L$ is said to be\emph{ positive}
if it admits a metric with positive curvature, i.e. such that the
curvature form $\omega$ is Kähler. Fixing a reference metric $\left\Vert \cdot\right\Vert $on
$L$ with curvature form $\omega_{0}$ any other metric on $L$ may
be expressed as $\left\Vert \cdot\right\Vert e^{-u/2},$ for $u\in C^{\infty}(X)$
and its curvature is positive iff $u\in\mathcal{H}(X,\omega).$ When
$L$ is the canonical line bundle $K_{X},$ i.e. the top exterior
power of the holomorphic cotangent bundle of $X:$ 
\[
K_{X}:=\det(T^{*}X)
\]
any volume form $dV$ on $X$ induces a smooth metric $\left\Vert \cdot\right\Vert _{dV}$
on $K_{X},$ by locally setting $\left\Vert dz\right\Vert _{dV}:=dz/dV,$
where $dz:=dz_{1}\wedge\cdots\wedge dz_{n}$ in terms of local holomorphic
coordinates. When $dV$ is the volume form of a given Kähler metric
$\omega$ on $X,$ i.e. $dV=\omega^{n}/n!,$ then its curvature form
may be identified with minus the Ricci curvature of $\omega,$ i.e.
\begin{equation}
\mbox{\ensuremath{\mbox{Ric}}\ensuremath{\omega=}}-dd^{c}\log\frac{dV}{c_{n}dz\wedge d\bar{z}},\label{eq:def of ricci curv}
\end{equation}
 where $c_{n}dz\wedge d\bar{z}$ is a short hand for the local Euclidean
volume form $\frac{i}{2}dz_{1}\wedge d\bar{z}_{1}\wedge\cdots\wedge\frac{i}{2}dz_{n}\wedge d\bar{z}_{n}.$
By a slight abuse of notation we will also write $\mbox{Ric \ensuremath{(dV)} }$for
the right hand side in formula \ref{eq:def of ricci curv}.

\subsubsection{\label{sub:Twisted-K=0000E4hler-Einstein-metrics}Twisted Kähler-Einstein
metrics }

A Kähler metric $\omega_{\beta}$ is said to be a \emph{twisted Kähler-Einstein
metric} if it satisfies the twisted Kähler-Einstein equation 
\begin{equation}
\mbox{\ensuremath{\mbox{Ric}}\ensuremath{\omega}}=-\beta\omega+\eta,\label{eq:tw ke eq in text}
\end{equation}
where the form $\eta$ is called the\emph{ twisting form. }Since $\omega_{\beta}$
is Kähler the form $\eta$ is necessarily closed and $J-$invariant.
The corresponding equation at the level of cohomology classes is

\[
[\eta]=-c_{1}(K_{X})+\beta[\omega]
\]
Fixing, once and for all, a volume form $dV$ on $X$ gives the following
one-to-one corresponds between twisting forms $\eta$ and Kähler forms
$\omega_{0}$ solving the cohomological equation above: 
\begin{equation}
\eta:=\beta\omega_{0}+\ensuremath{\mbox{Ric}}\ensuremath{dV}.\label{eq:def of eta}
\end{equation}
The following lemma then follows directly from the expression \ref{eq:def of ricci curv}
for the Ricci curvature of a Kähler metric:
\begin{lem}
\label{lem:twisted}Let $X$ be a compact complex manifold endowed
with a $J-$invariant and closed form $\eta.$ Then a Kähler form
$\omega_{\beta}$ solves the corresponding twisted Kähler-Einstein
equation \ref{eq:tw ke eq in text} iff $\omega_{\beta}:=\omega_{0}+dd^{c}\varphi_{\beta}$
for a unique $\varphi_{\beta}\in\mathcal{H}(X,\omega)$ solving the
PDE
\begin{equation}
\omega_{\varphi}^{n}=e^{\beta\varphi}dV\label{eq:ma eq with beta in text}
\end{equation}
The celebrated Aubin-Yau theorem may now be formulated as follows:\end{lem}
\begin{thm}
(Aubin-Yau) \cite{au,y} Given a compact complex manifold $X,$ endowed
with a Kähler form $\omega_{0}$ and a volume form $dV,$ The PDE
\ref{eq:ma eq with beta in text} admits, for any positive number
$\beta,$ a unique solution $\varphi_{\beta}\in\mathcal{H}(X,\omega).$
Equivalently, given a $J-$invariant and closed form $\eta$ such
that the class $[\eta]+c_{1}(K_{X})$ is positive (i.e. contains a
Kähler form) there exists a unique Kähler metric $\omega_{\beta}$
solving the twisted Kähler-Einstein equation \ref{eq:tw ke eq in text}.\end{thm}
\begin{example}
\label{exa:ke when K_X pos}A complex manifold $X$ admits a Kähler-Einstein
metric with negative Ricci curvature iff $K_{X}$ is positive (and
the metric is unique). Indeed, if $K_{X}$ is positive then, by the
very definition of positivity, we can take $\omega_{0}:=-\ensuremath{\mbox{Ric}}\ensuremath{dV}$
for some volume form on $X,$ ensuring that $\eta=0$ above, with
$\beta=1$ (and the converse is trivial). \end{example}
\begin{rem}
When $n\geq2$ the equation is \ref{eq:tw ke eq in text} precisely
the trace-reversed formulation of Einstein's equations on $X$ (with
Euclidean signature): $-\beta$ is the cosmological constant and $\eta$
is the trace-reversed stress-energy tensor. Here we are only concerned
with the solutions which are Kähler metrics.
\end{rem}

\subsubsection{The projection operator $P_{\omega_{0}}$ to the space $PSH(X,\omega_{0})$ }

Next, we recall the definition of the operator $P$ introduced in
\cite{b-b} (which turns out to be related to the limit as $\beta\rightarrow\infty$
of the equations \ref{eq:ma eq with beta in text}). Given $u\in C^{0}(X)$
we set 
\begin{equation}
(Pu)(x):=\sup_{\varphi\in\mathcal{H}(X,\omega_{0})}\{\varphi(x):\,\,\,\varphi\leq u\}\label{eq:def of proj operator in khler case}
\end{equation}
which defines an operator 
\[
P:\,C^{0}(X)\rightarrow PSH(X,\omega_{0})
\]
from $C^{0}(X)$ to the space $PSH(X,\omega_{0})$ of all $\omega_{0}-$psh
functions on $X,$ i.e. all upper semi-continuous functions $\varphi$
in $L^{1}(X)$ such that $\omega_{\varphi}\geq0$ in the sense of
currents. In fact, the operator $P$ preserves $C^{0}(X)$ and hence
defines a projection operator from $C^{0}(X)$ onto $PSH(X,\omega_{0})\cap C^{0}(X).$
More generally, if $u\in C^{\infty}(X),$ the current $dd^{c}(Pu)$
has coefficients in $L_{loc}^{\infty},$ i.e. 
\begin{equation}
\omega_{Pu}\in L_{loc}^{\infty}\label{eq:reg of projection}
\end{equation}
In fact, as shown in \cite{berm9}, taking $dV$ in equation \ref{eq:ma eq with beta in text}
to be of the form $dV=e^{-\beta u}dV$ one has 
\[
\lim_{\beta\rightarrow\infty}\varphi_{\beta}=Pu
\]
 uniformly on $X$ and with a uniform upper bound on the corresponding
Kähler forms $\omega_{\varphi_{\beta}}.$

\subsubsection{The Monge-Ampère operator and the functionals $\mathcal{E}$ and
$\mathcal{F}$ }

The second order operator 
\begin{equation}
\varphi\mapsto MA(\varphi):=\omega_{\varphi}^{n},\label{eq:ma}
\end{equation}
 appearing in the equation \ref{eq:ma eq with beta in text}, is the
complex \emph{Monge-Ampère operator} (with respect to the reference
form $\omega_{0}).$ \footnote{The terminology stems from the fact that when $\omega_{0}=0$ (which
can always be locally arranged by shifting $\varphi)$ the density
of $MA(\varphi)$ is proportional to the determinant of the complex
Hessian $\partial\bar{\partial}\varphi.$} By Stokes theorem 
\[
\int_{X}\omega_{\varphi}^{n}=\int_{X}\omega_{0}^{n}:=V
\]
which is hence a positive number independent of $\varphi\in C^{\infty}(X).$
Up to a trivial scaling we may and will assume that $V=1.$ When $n=1$
the operator $MA$ may be identified with the Laplacian, but when
$n\geq2$ it is fully non-linear. The one-form on $C^{\infty}(X)$
defined by $MA$ is closed and hence admits a primitive $\mathcal{E},$
i.e. a functional on $C^{\infty}(X)$ whose differential is given
by 
\begin{equation}
d\mathcal{E}_{|\varphi}=MA(\varphi).\label{eq:def of energyfunc as primitive}
\end{equation}
The functional $\mathcal{E}$ is only determined up to an additive
constant which may be fixed by the normalization condition $\mathcal{E}(0)=0.$

Using pluripotential theory \cite{begz,bbgz} the operator $MA$ can
be extended from $\mathcal{H}_{\omega_{0}}$ to all of $PSH(X,\omega_{0})$
giving a positive measures satisfying 
\[
\int_{X}MA(\varphi)\leq V(:=1)
\]
Similarly the functional $\mathcal{E}$ also extends from $\mathcal{H}(X,\omega_{0})$
to an increasing lower-semi continuous functional on $PSH(X,\omega_{0}).$
We then set 
\begin{equation}
\mathcal{F}(u):=(\mathcal{E}\circ P)(u),\label{eq:def of F as energy comp P}
\end{equation}
 which by \cite{b-b} defines a Gateaux differentiable functional
on $C^{0}(X).$ More precisely, \emph{
\begin{equation}
(d\mathcal{F})_{|u}=MA(Pu).\label{eq:formula for diff}
\end{equation}
}This setup leads to a direct variational approach for solving complex
Monge-Ampère equations, including the Aubin-Yau equation \ref{eq:ma eq with beta in text},
in the more general setting of big cohomology classes and singular
volume forms $dV$ \cite{bbgz} (compare Section \ref{sub:The-generalization-to big}).
However, in the present setting where $L$ is positive the pluripotential
theory can be dispensed with by observing that $MA(\varphi)$ is a
well-defined probability measure as long as $\omega_{\varphi}$ is
in $L_{loc}^{\infty}$ (using that $MA(\varphi)$ is point-wise defined
almost everywhere on $X$ ). Then $\mathcal{F}(u)$ may be defined
by first taking $u$ to be in $C^{\infty}$and using the regularity
result \ref{eq:reg of projection} for $Pu.$ One then defines $\mathcal{F}$
on $C^{0}(X)$ as the unique continuous extension of $\mathcal{F}$
from $C^{\infty}(X),$ using that $\mathcal{F}(u)$ is Lipschitz continuous
on $C^{\infty}(X)$ with respect to the $C^{0}-$norm (as follows
form general principles; see the beginning of the proof of Lemma \ref{lem:R-F duality etc}).

\subsection{The ``temperature deformed'' determinantal point processes on $X$}

Let $(X,L)$ be a polarized manifold, i.e. an $n-$dimensional complex
compact manifold $X$ endowed with a positive holomorphic line bundle
$L.$ We will denote by $H^{0}(X,kL)$ the space of all global holomorphic
sections with values in the $k$ th tensor power of $L$ (using additive
notation for tensor powers). By the Hilbert-Samuel theorem 
\[
N_{k}:=\dim H^{0}(X,kL)=Vk^{n}+o(k^{n}),
\]
where $V=\int_{X}c_{1}(L)^{n}>0.$

To the data $(\left\Vert \cdot\right\Vert ,dV,\beta_{k})$ consisting
of a Hermitian metric $\left\Vert \cdot\right\Vert $ on $L,$ a volume
form $dV$ on $X$ and a sequence of positive number $\beta_{k}$
we can associate the following sequence of symmetric probability measures
on $X^{N_{k}}:$ 
\begin{equation}
\mu^{(N_{k},\beta)}:=\frac{\left\Vert (\det S^{(k)})(x_{1},x_{2},...x_{N_{k}})\right\Vert ^{2\beta_{k}/k}dV^{\otimes N_{k}}}{Z_{N_{k},\beta}}\label{eq:prob measure def det text}
\end{equation}
where $\det S^{(k)}$ is a generator of the top exterior power $\Lambda^{N_{k}}H^{0}(X,kL),$
viewed as a one-dimensional subspace of $H^{0}(X^{N_{k}},(kL)^{\boxtimes N_{k}})$
under the usual isomorphism between $H^{0}(X^{N_{k}},(kL)^{\boxtimes N_{k}})$
and the $n$ fold tensor product of $H^{0}(X,L).$ The number $Z_{N_{k},\beta}$
is the normalizing constant 
\begin{equation}
Z_{N_{k},\beta}:=\int_{X^{N_{k}}}\left\Vert \det S^{(k)}\right\Vert ^{2\beta/k}dV^{\otimes N_{k}}\label{eq:def of partion function in beta-setting text}
\end{equation}
By homogeneity the probability measure $\mu^{(N_{k},\beta)}$ is independent
of the choice of generator $\det S^{(k)}$ and thus only depends on
the data $(\left\Vert \cdot\right\Vert ,dV,\beta_{k}).$ We will refer
to to the corresponding random point processes on $X,$ as the\emph{
temperature deformed determinantal point processes} on $X$ attached
to $(\left\Vert \cdot\right\Vert ,dV,\beta_{k})$ (the special case
$\beta_{k}=k$ defines a bona fide determinantal point process, as
recalled below). 
\begin{rem}
\label{rem:invariance}Since the transformation $(\left\Vert \cdot\right\Vert ,dV,\beta_{k})\mapsto(\left\Vert \cdot\right\Vert e^{-u/2},e^{u\beta_{k}}dV,\beta_{k}),$
for $u\in C^{0}(X),$ leaves the probability measure \ref{eq:prob measure def det text}
invariant, the processes above only depend on the data $(\left\Vert \cdot\right\Vert ,dV,\beta_{k})$
through the corresponding two form $\eta,$ defined by formula \ref{eq:def of eta}.
Moreover, to any twisting form $\eta$ such that the cohomology class
$([\eta]+c_{1}(K_{X}))/\beta_{k}$ defines a positive class in $H^{2}(X,\R)\cap H^{2}(X,\Z),$
i.e. is the first Chern class of a line bundle $L,$ arises from a
choice of data $(\left\Vert \cdot\right\Vert ,dV,\beta_{k})$ (compare
Section \ref{sub:Twisted-K=0000E4hler-Einstein-metrics}). 
\end{rem}
It will be convenient to take $\det S^{(k)}$ to be the generator
determined by a basis $s_{1},...,s_{N_{k}}$ in $H^{0}(X,kL)$ which
is orthonormal with respect to the $L^{2}-$product determined by
$(\left\Vert \cdot\right\Vert ,dV)$ for any fixed volume form $dV$
on $X:$ 
\[
\left\langle s,s\right\rangle _{L^{2}}:=\int_{X}\left\Vert s\right\Vert ^{2}dV
\]
We then take $(\det S^{(k)})(x_{1},x_{2},...,x_{N}):=$ 
\begin{equation}
=\det(s_{i}(x_{j})):=\sum_{\sigma\in S_{N_{k}}}(-1)^{\mbox{sign\ensuremath{(\sigma)}}}s_{1}(x_{\sigma(1)})\cdots s_{N_{k}}(x_{\sigma(N_{k})})\label{eq:def of det S}
\end{equation}

\begin{example}
\label{exa:proj subm}The model case of a polarized manifold is $(X,L)=(\P^{m},\mathcal{O}(1)),$
where $\P^{m}(:=\C^{m}-\{0\})/\C^{*}$ is $m-$dimensional complex
projective space and $\mathcal{O}(1)$ is the hyperplane line bundle
over $\P^{m}$ (the model positively curved metric on $\mathcal{O}(1)$
is the Fubini-Study metric induced from the Euclidean metric on $\C^{m}$).
More generally, taking $X$ to be a non-singular algebraic variety
of $\P^{m}$ and $L$ as the restriction to $X$ of $\mathcal{O}(1)$
gives a polarized where the elements in $H^{0}(X,kL)$ are, for $k$
sufficiently large, the restrictions to $X$ of homogeneous polynomials
of degree $k$ on $\P^{m}$ (in fact, by the Kodaira embedding theorem
any polarized manifold $(X,L)$ may, after replacing $L$ with a sufficiently
high tensor power, be concretely realized as $(X,\mathcal{O}(1)_{|X})).$
In the case of $X=\P^{1}$ (=the Riemann sphere) with $\left\Vert \cdot\right\Vert $
denoting the Fubini-Study metric on $\mathcal{O}(1)$ whose curvature
form $\omega_{0}$ is the invariant measure on $\P^{1}$ one can take
the base $\{s_{i}\}$ to consist of monomials and factorize 
\[
\left\Vert \det S^{(k)}\right\Vert (x_{1},x_{2},...,x_{N})=Z_{N}\prod_{1\leq i<j<N}|x_{i}-x_{j}|,
\]
where $N=k+1$ and $X$ has been identified with the unit-sphere in
Euclidean $\R^{3}$ and where $Z_{N}=N^{N}\binom{N-1}{0}...\binom{N-1}{N-1}/N!.$
In the physics literature the corresponding ensemble appears as a
\emph{Coulomb gas} of $N$ unit-charge particles (i.e a one component
plasma) confined to the sphere in a neutralizing uniform background
$\omega$ (see for example \cite{ca}). More generally, on any Riemann
surface of genus $g$ the bosonization formula \cite{b-v-} gives
\begin{equation}
\left\Vert \det S^{(k)}\right\Vert (x_{1},...x_{N})=Z_{N}\exp\left(-\sum_{i\neq j}G(x_{i},x_{j})+r(x_{1},....,x_{N})\right)\label{eq:boson formula}
\end{equation}
where $G$ is the Green function of the Laplacian induced by the metric
$\omega_{0}$ and where the second term $r$ appearing above vanishes
for genus $g=0,$ while for $g>0$ it may be expressed in terms of
the Riemann eta function on the Jacobian torus of the Riemann surface
$X$ (giving a contribution which is lower order than the first term;
see \cite{z2} and references therein). However, when $n>1$ it should
be stressed that there is no tractable formula for $\left\Vert \det S^{(k)}\right\Vert (x_{1},...x_{N}),$
even to the leading order.
\end{example}
When $\beta_{k}=k$ the probability measure $\mu^{(N_{k},\beta_{k})}$
in formula \ref{eq:prob measure def det text}defines a \emph{determinantal
point process} i.e. its density can be written as 
\[
\left\Vert \det_{i,j\leq N}(K^{(k)}(x_{i},x_{j}))\right\Vert /N_{k}!,
\]
 where $K^{(k)}(x,y)$ denotes the kernel of the orthogonal projection
onto the space $H^{0}(X,kL)$ viewed as a subspace of the space $C^{\infty}(X,kL)$
of all smooth sections equipped with the $L^{2}-$norm determined
by $(\left\Vert \cdot\right\Vert ,dV)$ \cite{h-k-p,berm 1 komma 5}. 

The following result generalizes the LDP in \cite{berm 1 komma 5}
for determinantal point processes (or more generally for the case
$\beta=\infty)$ to the general case where $\beta_{k}\rightarrow\beta\in]0,\infty]:$ 
\begin{thm}
\label{thm:def determ}Let $(X,L)$ be a polarized manifold and assume
given the data $(\left\Vert \cdot\right\Vert ,dV,\beta_{k})$ consisting
of a Hermitian metric $\left\Vert \cdot\right\Vert $ on $L,$ a volume
form $dV$ on $X$ and a sequence of positive number $\beta_{k}\rightarrow\beta\in]0,\infty].$
Then the law of the empirical measures $\delta_{N_{k}}$ of the corresponding
deformed determinantal point processes with $N_{k}$ particles satisfies
a LDP with speed $\beta_{k}N_{k}$ and rate functional 
\[
F_{\beta}(\mu)=E_{\omega_{0}}(\mu)+\frac{1}{\beta}D_{dV}(\mu)-C_{\beta},
\]
 where $E_{\omega_{0}}(\mu)$ is the pluricomplex energy of $\mu$
with respect to the curvature form $\omega_{0}$ of $\left\Vert \cdot\right\Vert $
and 
\[
C_{\beta}=\inf_{\mathcal{M}_{1}(X)}F_{\beta}=-\lim_{N\rightarrow\infty}\frac{1}{N_{k}\beta_{k}}\log Z_{N.\beta_{k}},
\]
 In particular, $\delta_{N_{k}}$ converges in law to the deterministic
measure given by the unique minimizer $\mu_{\beta}$ of $F_{\beta}.$
Moreover, when $\beta<\infty$ the measure $\mu_{\beta}$ is the normalized
volume form $\omega_{\beta}$ of the twisted Kähler-Einstein metric
corresponding to the twisting form $\eta:=\beta\omega_{0}+\ensuremath{\mbox{Ric}}\ensuremath{dV}.$
\end{thm}
In fact, the Kähler form $\omega_{\beta}$ may be recovered directly
from the limiting volume form $\mu_{\beta}$ by differentiation twice
(as follow from the very definition of the twisted Kähler-Einstein
equation \ref{eq:tw ke eq in text}): 
\[
\omega_{\beta}:=\frac{i}{2\pi}\frac{1}{\beta}\partial\bar{\partial}\log\frac{\mu_{\beta}}{dV}+\omega_{0},
\]
Using basic compactness properties of the space $PSH(X,\omega_{0})$
one then arrives at the following corollary (see \cite{berm8} for
the proof):
\begin{cor}
\label{cor:tw ke text}Given data as in the previous theorem with
$\beta\in]0,\infty[,$ the following sequence of Kähler forms on $X$
\[
\omega^{(k)}:=dd^{c}\frac{1}{\beta}\log\frac{\int_{X^{N_{k}-1}}\left\Vert \det S^{(k)}(\cdot,x_{2},...x_{N_{k}})\right\Vert ^{2\beta/k}dV^{\otimes(N_{k}-1)}}{dV}+\omega_{0},
\]
 converges to the unique solution $\omega_{\beta}$ of the the twisted
Kähler-Einstein metric corresponding to the twisting form $\eta:=\beta\omega_{0}+\ensuremath{\mbox{Ric}}\ensuremath{dV}.$\end{cor}
\begin{rem}
\label{rem:quasi-explicit}The previous corollary yields a quasi-explicit
way of approximating the solution $\omega_{\beta}$ to the twisted
KE equation in question (or equivalently the solution $\varphi_{\beta}$
of the corresponding complex Monge-Ampère equation \ref{eq:ma eq with beta in text}),
by performing integrals over the spaces $X^{N_{k}-1}$ of increasing
dimension. The procedure becomes explicit as soon as one has constructed
bases in the spaces $H^{0}(X,kL),$ for $k$ sufficiently large. 
\end{rem}

\subsubsection{The canonical random point processes on $X$ }

We start by recalling the basic fact that, by the very definition
of the canonical line bundle $K_{X},$ any holomorphic section $s_{k}$
of the $k$ th tensor power of $K_{X}$ (i.e. $s_{k}\in H^{0}(X,kK_{X}$)
induces a measure on $X,$ symbolically denoted by $(s_{k}\wedge\bar{s}_{k})^{1/k}.$
Concretely, given an open set $U\subset X$ with holomorphic coordinates
$(z_{1},...,z_{n})$ and writing $s_{k|U}=f_{k}dz^{\otimes k}$ for
a holomorphic function $f_{k}$ on $U,$ where $dz:=dz_{1}\wedge\cdots\wedge dz_{n}$
trivializes $K_{X}$ over $U,$ 
\[
(s_{k}\wedge\bar{s}_{k})_{|U}^{1/k}=|f_{k}|^{2/k}i^{n^{2}}dz\wedge d\bar{z},
\]
 which is independent of $U$ and thus defines a global measure on
$X$ (using any holomorphic atlas on $X).$ We also recall that any
volume form $dV$ on $X$ induces a metric $\left\Vert \cdot\right\Vert _{dV}$
on the canonical line bundle $K_{X}$ with the property that, if $s_{k}\in H^{0}(X,kK_{X})$
then $(s_{k}\wedge\bar{s}_{k})_{|U}^{1/k}$ may be expressed as 
\begin{equation}
(s_{k}\wedge\bar{s}_{k})_{|U}^{1/k}=\left\Vert s_{k}\right\Vert _{dV}^{2/k}dV,\label{eq:measure for s in terms of metric}
\end{equation}
 as follows immediately from the definitions.

Now, fixing a volume form $dV$ on $X$ we can apply the relation
\ref{eq:measure for s in terms of metric} to $X^{N}$ equipped with
the induced volume form $dV^{\otimes N}$ and the corresponding metric
on $L$ and deduce that the canonical probability measure $\mu^{(N_{k})}$
on $X^{N_{k}}$defined by formula \ref{eq:canon prob measure intro}
coincides with the probability measured in formula \ref{eq:prob measure def det text}
corresponding to the data $(\left\Vert \cdot\right\Vert _{dV},dV,1)$
Hence, Theorem \ref{thm:ke intro} is indeed a special case of Theorem
\ref{thm:def determ} (also using that $\eta=0$ for this particular
data).

\subsection{Proof of Theorem \ref{thm:def determ}}

To apply Theorem \ref{thm:gibbs intro} in the present setting first
note that the Hamiltonian is given by 
\begin{equation}
E^{(N_{k})}(x_{1},x_{2},...x_{N_{k}}):=-\frac{1}{k}\log\left\Vert (\det S^{(k)})(x_{1},x_{2},...x_{N_{k}})\right\Vert ^{2},\label{eq:def of Hamtiltonian in complex setting text}
\end{equation}
 where $\det S^{(k)}$ is defined by formula \ref{eq:def of det S}.
The validity of the first assumption in Theorem \ref{thm:gibbs intro}
is then a consequence of the following result from \cite{b-b}, where
$\beta_{N_{k}}=k:$ 
\begin{thm}
\label{thm A b-b} \cite{b-b}. Let $L\rightarrow X$ be a positive
line bundle equipped with a smooth Hermitian metric $\left\Vert \cdot\right\Vert $
on $L$ with curvature form $\omega_{0}$ and $dV$ a volume form
on $X.$ Then \emph{
\[
\lim_{k\rightarrow\infty}-\frac{1}{kN_{k}}\left(\log\int_{X^{N_{k}}}\left\Vert \det S^{(k)}\right\Vert ^{2}(x_{1},...,x_{N})e^{-ku(x_{1})-\cdots-ku(x_{N})}\right)=\mathcal{F}(u),
\]
where $\mathcal{F}$ is the Gateaux differentiable functional defined
by formula }\ref{eq:def of F as energy comp P}
\end{thm}
To verify the second assumption in Theorem \ref{thm:gibbs intro},
concerning quasi-superharmonicity, we first observe that we may as
well assume that $dV$ is the volume form $dV_{g}$ of the metric
$g$ defined by the Kähler form $\omega_{0}.$ Indeed, $dV=e^{-u\beta}dV_{g}$
for some smooth function $u$ and hence changing $dV$ corresponds
to changing the metric $\left\Vert \cdot\right\Vert $ to $\left\Vert \cdot\right\Vert e^{-u/2}.$
Next, we recall that, in general, $\log\left\Vert s\right\Vert ^{2}$
is $k\omega-$psh for any holomorphic section $s$ of $kL\rightarrow X$
(where $\omega$ is the curvature form of $\left\Vert \cdot\right\Vert $).
Hence, we get, 
\[
\Delta_{g}\log\left\Vert s\right\Vert ^{2/k}\geq-\lambda
\]
 for some positive constant $\lambda.$ Applying the latter inequality
to $\left\Vert \det(s^{(k)}(\cdot,x_{2},...,x_{N})\right\Vert $ for
$x_{2},...,x_{N}$ thus shows that Theorem \ref{thm:gibbs intro}
can be applied to get the LDP in Theorem \ref{thm:def determ}. 

Next, we will show that the unique minimizer $\mu_{\beta}$ of the
rate functional $F_{\beta}$ appearing in Theorem \ref{thm:gibbs intro}
coincides with the normalized volume form $\omega_{\beta}$ of the
corresponding twisted Kähler-Einstein metric, by applying the general
Lemma \ref{lem:R-F duality etc}. It should however be stressed that
while the infimum in the left hand side of formula\ref{eq:inf is sup in lemma}
is always attained at some $\mu_{0}\in\mathcal{M}_{1}(X)$ (by weak
compactness and lower-semi continuity) this is not so for the right
hand side, in general. But in the present setting the sup is attained,
when $L$ is assumed to be positive, thanks to the Aubin-Yau theorem.
Indeed, first setting 
\[
g(u)=\beta^{-1}\log\int e^{\beta u}dV,
\]
 for a given $\beta\in]0,\infty[$ gives $g^{*}(\mu)=\beta^{-1}D_{dV}(\mu)$
if $\mu\in\mathcal{M}_{1}(X)$ and $g^{*}(\mu)=\infty$ otherwise,
as is well-known \cite{d-z} (and follows from Jensen's inequality
applied to the log). Moreover, by the dominated convergence theorem
\[
dg_{|u}=\frac{e^{\beta u}dV}{\int_{X}e^{\beta u}dV}\in\mathcal{M}_{1}(X)
\]
 Letting $\mathcal{F}$ be the functional on $C^{0}(X)$ defined by
formula \ref{eq:def of F as energy comp P} the critical point equation
\ref{eq:crit pt eq} thus becomes 
\[
MA(Pu)=\frac{e^{\beta u}dV}{\int_{X}e^{\beta u}dV},
\]
when $u$ is smooth, say. Up to replacing $u$ by $u+C$ we may as
well assume that the denominator above is equal to $1.$ In particular,
when $u\in\mathcal{H}(X,\omega)$ the equation above is precisely
the Aubin-Yau equation \ref{eq:ma eq with beta in text}, which, by
the Aubin-Yau theorem admits a (unique) solution $u_{\beta}\in\mathcal{H}(X,\omega).$
Hence, by the previous lemma $\mu_{\beta}:=MA(u_{\beta})$ is the
unique minimizer of the rate functional $F_{\beta}$ appearing in
Theorem \ref{thm:gibbs intro}, in the present setting. Finally, as
explained in Section \ref{sub:Twisted-K=0000E4hler-Einstein-metrics}
$\mu_{\beta}$ is the volume form of the Kähler form $\omega_{\beta}$
solving the twisted Kähler-Einstein equation \ref{eq:tw ke eq in text}. 
\begin{rem}
\label{rem:energy}To see the relation to the pluricomplex energy
introduced in \cite{bbgz} we write, as in formula \ref{eq:proof of lemma inf is sup},
\[
f^{*}(\mu)=\sup_{u\in C^{0}(X)}\mathcal{E}(Pu)-\left\langle u,\mu\right\rangle ,
\]
when $\mu\in\mathcal{M}_{1}(X),$ which coincides with the\emph{ pluricomplex
energy }of $\mu,$ with respect to $\omega_{0}$ in \cite{bbgz} (using
the notation in \cite{berm6}). More concretely, a direct calculation
reveals that when $\mu$ is a volume form 
\begin{equation}
E(\mu)=\frac{1}{V}\sum_{j=0}^{n-1}\frac{1}{j+2}\int_{X}d\varphi_{\mu}\wedge d^{c}\varphi_{\mu}\wedge\frac{(dd^{c}\varphi_{\mu}+\omega_{0})^{j}}{j!}\wedge\frac{\omega_{0}^{n-1-j}}{(n-1-j)!},\label{eq:energy in terms of potential}
\end{equation}
where $\varphi_{\mu}\in\mathcal{H}(X,\omega_{0})$ is the solution
to the Calabi-Yau equation \ref{eq:c-y eq intro}, which in Aubin's
notation \cite{au} means that $E(\mu)=c_{n}(I-J)(\varphi_{\mu})$
(using \cite{begz} the formula above holds for any $\mu$ such that
$E(\mu)<\infty,$ by letting $\wedge$ denote the non-pluripolar products
\cite{begz}). Thus $E(\mu)$ is a generalization of the classical
Dirichlet energy on a Riemann surface. The relation $F_{\beta}(\omega^{n})=\mbox{\ensuremath{\kappa}(\ensuremath{\omega}), }$where
$\kappa$ denotes the twisted version of \emph{Mabuchi's K-energy}
then follows from the Chen-Tian formula for the K-energy (see \cite{berm6}
and \cite{berm8} for a direct proof using convex analysis). Moreover,
the restriction to $\mathcal{H}(X,\omega_{0})$ of the dual functional
$f(-u)+g(u)$ appearing in Lemma \ref{lem:R-F duality etc} coincides
with the \emph{Ding functional} in Kähler geometry \cite{berm6}.
An alternative proof of the fact that $\omega_{\beta}^{n}$ minimizes
$F_{\beta}$ on $\mathcal{M}_{1}(X)$ can then be given by using that
$\omega_{\beta}$ is a critical point of $\kappa$ and hence, by convexity,
minimizes $\kappa$ on $\mathcal{H}(X,\omega_{0}).$ Accordingly,
the Calabi-Yau isomorphism $\omega\mapsto\omega^{n}$ shows that $\omega_{\beta}^{n}$
minimizes the restriction of $F_{\beta}$ to the subspace of all volume
forms in $\mathcal{M}_{1}(X).$ However, showing that the infimum
of $F_{\beta}$ over all of $\mathcal{M}_{1}(X)$ coincides with the
infimum over the subspace of volume forms requires the following non-trivial
fact: any $\mu$ such that $E(\mu)<\infty$ can be written as a weak
limit of volume forms $\mu_{j}$ such that $E(\mu_{j})\rightarrow E(\mu)$
and $D_{dV}(\mu_{j})\rightarrow D_{dV}(\mu)$ (see \cite{bdl} where
more general results are obtained).
\end{rem}

\subsection{\label{sub:The-generalization-to big}The generalization to big line
bundles and varieties of positive Kodaira dimension}

Let us briefly give some indications about the extension of Theorem
\ref{thm:def determ} to line bundles $L$ which are merely assumed
big, established in the companion paper \cite{berm8}. In analytic
terms $L$ is big iff $c_{1}(L)$ contains a positive current on $X$
which is strictly positive in the sense that it is bounded from below
by a Kähler form. However, in general, there is a proper open subset
$\Omega\subset X$ such that\emph{ all }positive currents in $c_{1}(L)$
are equal to $-\infty$ on the complement $X-\Omega$ (which can be
taken to be a complex subvariety of $X).$ Fixing a reference smooth
Hermitian metric $\left\Vert \cdot\right\Vert $ on $L$ with curvature
form $\omega_{0}$ in $c_{1}(L)$ the space of positive currents in
$c_{1}(L)$ gets identified, as before, with the space $PSH(X,\omega_{0})$
of all $\omega_{0}-$psh functions, modulo constants (however, in
general all elements in $PSH(X,\omega_{0})$ will be singular along
the subvariety $X-\Omega$). Moreover, the non-pluripolar Monge-Ampère
operator can be defined on $PSH(X,\omega_{0}),$ by restricting to
$\Omega$ \cite{begz}. Then the functional $\mathcal{F}$ can be
defined essentially as before and Theorem\ref{thm:def determ} still
holds (again using \cite{b-b} to verify the first assumption in Theorem
\ref{thm:gibbs intro}) Invoking, the general Theorem \ref{thm:gibbs intro}
thus establishes an LDP with a rate functional $F_{\beta},$ admitting
a unique minimizer $\mu_{\beta}$ as before. However, one new difficulty
is to show that $\mu_{\beta}$ can be written as $MA(\varphi_{\beta})$
for the solution to the equation \ref{eq:ma eq with beta in text}
with minimal singularities, whose existence is provided by the general
results in \cite{begz,bbgz}. The problem is that Lemma \ref{lem:R-F duality etc}
cannot be applied as it is not clear that $\varphi_{\beta}$ is of
the form $Pu$ for some $u$ in $C^{0}(X)$ (even if $u$ can be taken
to be in $L^{\infty}(X)$). But using the variational calculus in
\cite{bbgz,berm6} shows that $\mu_{\beta}$ is of the desired form. 

In particular, when $K_{X}$ is big, i.e. $X$ is a variety of general
type, the corresponding positive current $\omega_{\beta}$ is the
canonical Kähler-Einstein current in $X$ \cite{begz,bbgz}. In the
general case of a variety of positive Kodaira dimension $\kappa\leq n$
(where $\kappa=n$ iff $K_{X}$ is big) one can use the Ithaka fibration
$X\rightarrow Y$ to represent $K_{X}$ as the pull-back of a big
line bundle $L$ on the $\kappa-$dimensional manifold $Y.$ Using
the Fujino-Mori canonical bundle formula this reduces the proof of
the convergence on $X$ to the application of a generalization of
Theorem \ref{thm:def determ} concerning big line bundles on $Y$
endowed with a singular volume form $dV.$ As shown in \cite{berm8}
this realizes the corresponding canonical limiting current $\omega_{\beta}$
as the pull-back to $X$ of a (singular) Kähler form on $Y$ solving
a twisted Kähler-Einstein equation of the form \ref{eq:tw ke eq in text},
where $\eta$ is a current on $Y$ determined by the geometry of $X$
(the canonical current $\omega_{\beta}$ first appeared in a different
geometric context in \cite{s-t,ts}).

\section{\label{sec:Outlook}Outlook}

\subsection{$\beta=0$}

Let $(X,L)$ be a polarized manifold and fix a Kähler metric $\omega_{0}$
in $c_{1}(L).$ By Corollary \ref{cor:tw ke text} (and well-known
stability properties of the complex Monge-Ampère operator) one can
recover the unique (normalized) smooth solution to the Calabi-Yau
equation 
\begin{equation}
(\omega_{0}+i\partial\bar{\partial}\varphi)^{n}=dV,\label{eq:c-y eq in outlo}
\end{equation}
 \cite{y} as the double limit 
\[
\varphi:=\lim_{\beta\rightarrow\infty}\lim_{k\rightarrow\infty}\varphi_{\beta}^{(k)},\,\,\,\varphi_{\beta}^{(k)}:=\frac{1}{\beta}\log\frac{\int_{X^{N_{k}-1}}\left\Vert \det S^{(k)}(\cdot,x_{2},...x_{N_{k}})\right\Vert ^{2\beta/k}dV^{\otimes(N_{k}-1)}}{dV}-\log Z_{N}
\]
Formally interchanging the two limits thus suggests the following 
\begin{conjecture}
\label{conj:c-y eq}Let $(X,L)$ be a polarized manifold and $\omega_{0}$
a Kähler metric in $c_{1}(L).$ Then the unique smooth solution $\varphi$
to the Calabi-Yau equation \ref{eq:c-y eq in outlo}, normalized so
that $\int_{X}\varphi dV=0,$ may be represented as the following
limit in $L^{1}(X):$
\[
\varphi:=\lim_{k\rightarrow\infty}\varphi^{(k)},\,\,\,\varphi^{(k)}:=\frac{1}{k}\frac{\int_{X^{N_{k}-1}}\log\left\Vert \det S^{(k)}(\cdot,x_{2},...x_{N_{k}})\right\Vert ^{2}dV^{\otimes(N_{k}-1)}}{dV}-C_{k},
\]
 where the constant $C_{k}$ ensures that $\int_{X}\varphi^{(k)}dV=0.$ 
\end{conjecture}
The conjectural formula above can be seen as a generalization to the
non-linear complex Monge-Ampère operator of the classical Green's
formula for the solution of the Poisson equation for the Laplacian\textbf{
}on a Riemann surface. Indeed, when $X$ is a Riemann surface the
limit $\varphi$ above is precisely given by the Green formula in
question (as follows from the bosonization formula \ref{eq:boson formula}).
It turns out that the validity of the conjecture above would follow
from the existence of the corresponding mean energy $\bar{E}(\mu),$
for any volume form $\mu$ (see Problem \ref{prob:mean en}). This
is shown precisely as in the setting of the\emph{ real} Monge-Ampère
operator considered in \cite{berm7,hu} where the analog of the previous
conjecture was established using permanents as a replacements of the
determinants appearing in the present setting. In particular, when
$X$ is a Calabi-Yau manifold, i.e. $K_{X}$ is trivial, the conjecture
would imply a quasi-explicit formula for the unique Ricci flat Kähler
metric $\omega\in c_{1}(L),$ i.e. solving the Kähler-Einstein equation
with vanishing cosmological constant, $\Lambda=0.$

\subsection{$\beta<0$}

By Lemma \ref{lem:asympt for beta negative implies exist of mean}
the existence of the mean energy (and thus the resolution of the conjecture
above) would follow if one could establish the asymptotics in formula
\ref{eq:asympt in formula negative} of the corresponding partition
functions $Z_{N_{k},\beta/k}$ (assumed finite) for all $\beta>\beta_{0},$
for some negative number $\beta_{0}.$ It can be shown that $Z_{N,\beta_{N}}$
is indeed finite for for some negative $\beta_{0},$ sufficiently
close to zero. In fact, both sides of formula \ref{eq:asympt in formula negative}
are finite when $\beta>\beta_{0}$ (where the critical negative $\beta_{0}$
depends on $(X,L)).$ This motives the following
\begin{conjecture}
Let $(X,L)$ be a polarized manifold and assume given the data $(\left\Vert \cdot\right\Vert ,dV)$
consisting of a Hermitian metric $\left\Vert \cdot\right\Vert $ on
$L,$ a volume form $dV$ on $X.$ For a given negative number $\beta_{0}$
the following is equivalent: 
\begin{itemize}
\item For any $\beta>\beta_{0}$ the partition functions $Z_{N_{k},\beta}$
are finite for $k$ sufficiently large
\item For any $\beta>\beta_{0}$ the functional $\beta F_{\beta}$ admits
a minimizer on $\mathcal{M}_{1}(X)$
\item For any $\beta>\beta_{0}$ the measures $(\delta_{N})_{*}\left(e^{-\beta H^{(N_{k})}}dV^{\otimes N_{k}}\right)$
on $\mathcal{M}_{1}(X)$ satisfy a LDP with speed$N$ and rate functional
\[
\beta F_{\beta}(\mu)=\beta E_{\omega_{0}}(\mu)+D_{dV}(\mu)
\]
 where $E_{\omega_{0}}(\mu)$ is the pluricomplex energy of $\mu$
with respect to the curvature form $\omega_{0}$ of $\left\Vert \cdot\right\Vert .$ 
\end{itemize}
\end{conjecture}
In particular, if the conjectural LDP above holds then the functional
$\beta F_{\beta}$ is lower semi-continuous and the large $N-$limit
of the laws of $\delta_{N_{k}}$ for the corresponding random point
processes is concentrated on the (non-empty) set of minimizers of
$\beta F_{\beta}.$ By \cite{berm6} any such minimizer is the volume
form of a Kähler metric $\omega_{\beta}$ solving the twisted Kähler-Einstein
equation \ref{eq:tw ke eq in text} corresponding to the data $(\omega_{0},dV,\beta)$
and $\beta F_{\beta}$ may be identified with the corresponding twisted
K-energy functional. Moreover, if the LDP holds then it follows that
$Z_{N_{k},\beta}\leq C_{\beta}^{N},$ when $\beta>\beta_{0}.$ The
conjecture should be contrasted with the fact that, in general, $\beta F_{\beta}$
is unbounded from below if $\beta$ is sufficiently negative and even
when $\beta F_{\beta}$ is bounded from below there exist, in general,
twisted Kähler-Einstein metrics whose volume forms do not minimizer
of $\beta F_{\beta}.$ 

In the case when $L$ is the dual $-K_{X}$ of the canonical line
bundle, i.e. $X$ is a Fano manifold (which equivalently means that
$\eta$ can be taken to be zero) the equivalence between the two points
in the conjecture above can be seen as a probabilistic analog of the
Yau-Tian-Donaldson conjecture saying that a Fano manifold $X$ admits
a Kähler-Einstein metric with positive Ricci curvature ((i.e. $\Lambda>0)$
iff $X$ is K-stable in the algebro-geometric sense; see the companion
paper \cite{berm8} for more detailed explanations of these relations.

Interestingly, the notion of negative temperature has already appeared
in Onsager's work on the 2D vortex model \cite{o}. Using the bosonization
formula \ref{eq:boson formula} on a Riemann surface and large $N-$results
for vortex models (as in \cite{clmp,k,bo-g}) it can be shown that
the conjecture above holds when $X$ is a Riemann surface. Moreover,
then the critical $\beta_{0}=2$ when the volume (degree) of $L$
is normalized to be one, which in our normalizations corresponds to
the critical negative temperature in the vortex model (a detailed
proof will appear elsewhere).

\section{Appendix: the constant in the Cheng-Yau gradient estimate}

Set $\phi:=|\nabla u/u|$ and $F:=\phi(a^{2}-\rho^{2}).$ First, Bochner's
identity gives after some calculations \cite{c-y} that, for any $x,$
\begin{equation}
\frac{\Delta\phi}{\phi}\geq\frac{\phi^{2}}{(n-1)}-(n-1)k^{2}-(2-\frac{2}{(n-1)})\frac{\nabla\phi}{\phi}\cdot\frac{\nabla u}{u}\label{eq:Bochner}
\end{equation}
Let now $x_{1}$ be a point in the interior of $B_{a}(x_{0})$ where
$F$ attains it maximum and assume that $\rho(:=d(x,x_{0}))$ is smooth
close to $x_{1}.$ Next $\nabla F=0$ at $x_{1}$ gives 
\begin{equation}
\frac{\nabla\phi}{\phi}=\frac{\nabla\rho^{2}}{a^{2}-\rho^{2}}=\frac{2\rho\nabla\rho}{a^{2}-\rho^{2}}\label{eq:nabla f is 0}
\end{equation}
(in the following all (in-)equalities are evaluated at $x=x_{1})$
and $\Delta F\leq0$ at $x_{1}$ gives 
\[
\frac{\Delta\phi}{\phi}-\frac{\Delta\rho^{2}}{a^{2}-\rho^{2}}-\frac{2|\nabla\rho^{2}|^{2}}{(a^{2}-\rho^{2})^{2}}\leq0
\]
Now, by the Laplacian comparison 
\[
\Delta\rho^{2}\leq2+2(n-1)(1+k\rho^{2})
\]
Substituting this into the previous inequality we get (using $|\nabla\rho|\leq1)$
\begin{equation}
\frac{\Delta\phi}{\phi}-\frac{2+2(n-1)(1+k\rho^{2})}{a^{2}-\rho^{2}}-\frac{8\rho^{2}}{(a^{2}-\rho^{2})^{2}}\leq0\label{eq:conseq of laplace of f neg}
\end{equation}
By \ref{eq:nabla f is 0} 
\[
-\frac{\nabla\phi}{\phi}\cdot\frac{\nabla u}{u}\geq-\frac{2\rho\phi}{a^{2}-\rho^{2}}
\]
Hence, equation \ref{eq:Bochner} combined with equations \ref{eq:conseq of laplace of f neg}
and the previous inequality gives 
\[
0\geq\frac{\phi^{2}}{(n-1)}-(n-1)k^{2}-\frac{4(n-2)}{(n-1)}\frac{2\rho\phi}{a^{2}-\rho^{2}}-\frac{(2+2(n-1))(1+k\rho^{2})}{a^{2}-\rho^{2}}-\frac{8\rho^{2}}{(a^{2}-\rho^{2})^{2}},
\]
 Equivalently, multiplying by $(a^{2}-\rho^{2})^{2}$ gives 
\[
0\geq\frac{F^{2}}{(n-1)}-(n-1)k^{2}(a^{2}-\rho^{2})^{2}-\frac{4(n-2)}{(n-1)}2\rho F-(2+2(n-1))(1+k\rho^{2})(a^{2}-\rho^{2})-8\rho^{2},
\]
Since we are only interested in the large $n$ limit we deduce from
the previous inequality that 
\[
0\geq\frac{F^{2}}{(n-1)}-8\rho F-nk^{2}(a^{2}-\rho^{2})^{2}-2n(1+k\rho^{2})(a^{2}-\rho^{2})-8\rho^{2}
\]
 giving, after multiplication by $n,$ 
\[
0\geq F^{2}-8anF-n^{2}k^{2}(a^{2}-\rho^{2})^{2}-2n^{2}(1+ka^{2})(a^{2}-\rho^{2})-8a^{2}n,
\]
which we write as 
\[
(4an)^{2}+n^{2}k^{2}(a^{2}-\rho^{2})^{2}+2n^{2}(1+ka^{2})(a^{2}-\rho^{2})+8a^{2}n\geq(F-4an)^{2}
\]
giving 
\[
a^{2}\left(16n{}^{2}+n^{2}k^{2}a^{2}+2n^{2}(1+ka^{2})+8n\right)\geq(F-4an)^{2}
\]
Hence, 
\[
a^{2}n^{2}\left(26+3k^{2}a^{2}\right)\geq(F-4an)^{2}
\]
 giving 
\[
an\left(\left(26+3k^{2}a^{2}\right)^{1/2}+4\right)\geq F:=\phi(a-\rho)(a+\rho)\geq\phi(a-\phi)a,
\]
so that 
\[
n\left(\left(26+3k^{2}a^{2}\right)^{1/2}+4\right)\geq\phi(a-\phi),
\]
showing that there exists an absolute constant $C$ such that 
\[
Cn\left(1+ka\right)\geq\phi(a-\phi),
\]
as desired.

\end{document}